\documentclass[11pt]{article}

\newcommand{\blind}{0}

\usepackage{layout}

\usepackage{authblk} % for multiple author

\usepackage{amsmath, amsthm, amsfonts}  % ADD: amsthm
\usepackage{graphicx} %,psfrag,epsf
\usepackage{enumerate}
\usepackage[authoryear]{natbib}
\usepackage{url} % not crucial - just used below for the URL 

\usepackage{mathtools}   % for \norm
\usepackage{enumitem}   % for customized enumerate
\usepackage[table,xcdraw]{xcolor}   % for \textcolor
\usepackage{pdflscape}
\usepackage{afterpage}
\usepackage{capt-of}
\usepackage{adjustbox}

\usepackage[many]{tcolorbox}
\usepackage{multirow}   % for multirow in table

\usepackage{float}   % for putting table HERE
\usepackage{nameref}   % for label/ref a paragraph  
\usepackage{subcaption}  % for put two picture into one figure
\usepackage[titletoc,title]{appendix}  % Appendix A
\usepackage{longtable} % long table

\usepackage[font=normalsize]{caption}     %% make caption in normal size

\DeclarePairedDelimiterX{\norm}[1]{\lVert}{\rVert}{#1}   % need: \usepackage{mathtools}

\newcommand{\utwi}[1]{\mbox{\boldmath $ #1$}}

\newcommand{\ba}{{\utwi{a}}}

\newcommand{\bff}{{\utwi{f}}}

\newcommand{\bh}{{\utwi{h}}}

\newcommand{\bq}{{\utwi{q}}}

\newcommand{\bu}{{\utwi{u}}}
\newcommand{\bv}{{\utwi{v}}}
\newcommand{\bw}{{\utwi{w}}}
\newcommand{\bx}{{\utwi{x}}}
\newcommand{\by}{{\utwi{y}}}
\newcommand{\bz}{{\utwi{z}}}
\newcommand{\bA}{{\utwi{A}}}

\newcommand{\bC}{{\utwi{C}}}
\newcommand{\bD}{{\utwi{D}}}
\newcommand{\bE}{{\utwi{E}}}
\newcommand{\bF}{{\utwi{F}}}
\newcommand{\bG}{{\utwi{G}}}
\newcommand{\bH}{{\utwi{H}}}
\newcommand{\bI}{{\utwi{I}}}

\newcommand{\bK}{{\utwi{K}}}
\newcommand{\bL}{{\utwi{L}}}
\newcommand{\bM}{{\utwi{M}}}

\newcommand{\bO}{{\utwi{O}}}
\newcommand{\bP}{{\utwi{P}}}
\newcommand{\bQ}{{\utwi{Q}}}
\newcommand{\bR}{{\utwi{R}}}
\newcommand{\bS}{{\utwi{S}}}

\newcommand{\bU}{{\utwi{U}}}

\newcommand{\bW}{{\utwi{W}}}
\newcommand{\bX}{{\utwi{X}}}
\newcommand{\bY}{{\utwi{Y}}}
\newcommand{\bZ}{{\utwi{Z}}}

\newcommand{\bPsi}{{\utwi{\mathnormal\Psi}}}
\newcommand{\bPhi}{{\utwi{\mathnormal\Phi}}}
\newcommand{\bXi}{{\utwi{\mathnormal\Xi}}}

\newcommand{\blam}{{\utwi{\mathnormal\lambda}}}

\newcommand{\bepsilon}{{\utwi{\epsilon}}}
\newcommand{\bGamma}{{\utwi{\Gamma}}}
\newcommand{\bLambda}{{\utwi{\Lambda}}}
\newcommand{\bOmega}{{\utwi{\Omega}}}
\renewcommand{\bPhi}{{\utwi{\Phi}}}
\renewcommand{\bPsi}{{\utwi{\Psi}}}
\newcommand{\bSigma}{{\utwi{\Sigma}}}
\newcommand{\bTheta}{\utwi{\Theta}}

\newcommand{\bzero}{{\utwi{0}}}
\begin{document}

\def\spacingset#1{\renewcommand{\baselinestretch}%
{#1}\small\normalsize} \spacingset{1}

%%%%%%%%%%%%%%%%%%%%%%%%%%%%%%%%%%%%%%%%%%%%%%%%%%%%%%%%%%%%%%%%%%%%%%%%%%%%%%
\newtheorem{definition}{Definition}[section]
\newtheorem{theorem}{Theorem}
\newtheorem{lemma}{Lemma} 
% \newtheorem{lemma}[definition]{Lemma}
%newtheorem{corollary}[definition]{Corollary}
\newtheorem{corollary}{Corollary}
\newtheorem{proposition}{Proposition}
%%%
%%%
\newtheorem{remarkx}{Remark}
\newenvironment{remark}
{\begin{remarkx}
\em}
{\end{remarkx}}

%%%%%%%%%%%%%%%%%%%%%%%%%%%%%%%%%%%%%%%%%%%%%%%%%%%%%%%%%%%%%%%%%%%%%%%%%%%%%%
\if0\blind
{
  \title{\bf Constrained Factor Models for \\ High-Dimensional Matrix-Variate \\ Time Series}
  \author[1]{Elynn Y. Chen \thanks{Supported in part by NSF Grants DMS-1803241.}}
  \author[2]{Ruey S. Tsay \thanks{Supported in part by the Booth School of Business, University of Chicago.}}
  \author[3]{Rong Chen \thanks{Supported in part by NSF Grants DMS-1503409, DMS-1737857, DMS-1803241 and IIS-1741390.}}
  \affil[1]{Department of Operations Research and Financial Engineering, Princeton University}  
  \affil[2]{Booth School of Business, University of Chicago}
  \affil[3]{Department of Statistics and Biostatistics, Rutgers University}
  \date{\vspace{-5ex}}
  \maketitle
} \fi

\if1\blind
{
  \bigskip
  \bigskip
  \bigskip
  \title{\bf Constrained Factor Models for \\ High-Dimensional Matrix-Variate \\ Time Series}
  \date{\vspace{-5ex}}
  \maketitle
  \medskip
} \fi

\bigskip
\begin{abstract}
High-dimensional matrix-variate time series data are becoming widely available in many scientific fields, such as economics, biology and meteorology. 
To achieve significant dimension reduction while preserving the intrinsic matrix structure and temporal dynamics in such data, \cite{wang2018factor} proposed a matrix factor model that is shown to be able to provide effective analysis. 
In this paper, we establish a general framework for incorporating domain and prior knowledge in the matrix factor model through linear constraints. 
The proposed framework is shown to be useful in achieving parsimonious parameterization, facilitating interpretation of the latent matrix factor, and identifying specific factors of interest. 
Fully utilizing the prior-knowledge-induced constraints results in more efficient and accurate modeling, inference, dimension reduction as well as a clear and better interpretation of the results. 
Constrained, multi-term, and partially constrained factor models for matrix-variate time series are developed, with efficient estimation procedures and their asymptotic properties. 
We show that the convergence rates of the constrained factor loading matrices are much faster than those of the conventional matrix factor analysis under many situations. 
Simulation studies are carried out to demonstrate finite-sample performance of the proposed method and its associated asymptotic properties. 
We illustrate the proposed model with three applications, where the constrained matrix-factor models outperform their unconstrained counterparts in the power of variance explanation under the out-of-sample 10-fold cross-validation setting.
\end{abstract}

\vspace{3em}

\noindent%
{\it Keywords:}  Constrained eigen-analysis; Convergence in L2-norm; Dimension reduction; Factor model, Matrix-variate time series.

\newpage
%\spacingset{1.45} % DON'T change the spacing!
\spacingset{1.3}
\section{Introduction}
\label{sec:intro}

High-dimensional matrix-variate time series have been widely observed nowadays in a variety of scientific fields including economics, meteorology, and ecology. 
For example, the World Bank and the International Monetary Fund collect and publish macroeconomic data of more than thirty variables spanning over one hundred years and over two hundred countries covering a variety of demographic, social, political, and economic topics. 
These data neatly form a matrix-variate time series with rows representing the countries and columns representing various macroeconomic indexes. 
Typical factor analysis of such data either converts the matrix into a vector or modeling the row or column vectors separately \citep{Chamberlain-1983, Chamberlain-Rothschild-1983, Bai-2003, Bai-Ng-2002, Bai-Ng-2007, Forni-Hallin-Lippi-Reichline-2000, forni2004generalized, pan2008modelling, Lam-Yao-2011, Lam-Yao-2012, chang2015high}. 
However, the components of matrix-variates are often dependent among rows and columns with 
certain well-defined structure. 
Vectorizing a matrix-valued response, or modeling the row or column vectors separately may overlook some intrinsic dependency and fail to capture the matrix structure.
\cite{wang2018factor} propose a matrix factor model that maintains and utilizes the matrix structure of the data to achieve significant dimension reduction.

In factor analysis of matrix time series and in many other types of high-dimensional data, the problem of factor interpretations is of paramount importance. 
Furthermore, it is important in many practical applications to obtain specific latent factors related to certain domain theories, and with the aid of these specific factors to predict future values of interest more accurately. 
For example, financial researchers may be interested in extracting the latent factors of level, slope, and curvatures of the interest-rate yield curve and in predicting future equity prices based on those factors \citep{diebold2005modeling, diebold2006macroeconomy, rudebusch2008macro, bansal2014stock}. 

In many applications, relevant prior or domain knowledge is available or data themselves exhibit certain specific structure. Additional covariates may also be measured. For example, in business and economic forecasting, sector or group information of variables under study is often available. 
Such {\em a priori} information can be incorporated to improve the accuracy and inference of the analysis and to produce more parsimonious and interpretable factors. 
In other cases, the existing domain knowledge may intrigue researchers' interest in some specific factors. 
The theories and prior experience may provide guidance for specifying the measurable variables related to the specific factors of interest. 
It is then desirable to build proper constraints based on those measurable variables in order to effectively obtain the factors of interest. 
%It is then desirable to constrain the dimension of the factor representation in order to obtain effectively an adequate representation of the collected variables.

%This may, and often will, cut down considerably on the variance accounted for in the original data. This increase in residual error should be well paid for by the much crisper meaning of dimensions which are obtained and by the ease of predicting the locations of new objects in the space from knowledge of their values on the outside variables. 

To address these important issues and practical needs, we extend the matrix factor model of \cite{wang2018factor} by imposing natural constraints among the column and row variables to incorporate prior knowledge or to induce specific factors. Incorporating {\em a priori} information in parameter estimation has been widely used in statistical analysis, such as the constrained maximum likelihood estimation, constrained least squares, and penalized least squares. Constrained maximum likelihood estimation with the parameter space defined by linear or smooth nonlinear constraints have been explored in the literature.  \cite{hathaway1985constrained} applies the constrained maximum likelihood estimation to the problem of mixture normal distributions and shows that the constrained estimation avoids the problems of singularities and spurious maximizers often encountered by an unconstrained estimation. \cite{geyer1991constrained} proposes a general approach applicable to many models specified by constraints on the parameter space and illustrates his approach with a constrained logistic regression of the incidence of Down's syndrome on maternal age. Penalty methods have also been customarily used to enforce constraints in statistical models including generalized linear models, generalized estimating equations, proportional hazards models, and M-estimators. See, for example,  \cite{frank1993statistical}, \cite{tibshirani1996regression}, \cite{liu2007support}, \cite{fan2001variable}, \cite{zou2006adaptive}, and \cite{zhang2007adaptive}.  
It is shown that including the soft constraints as penalizing term enhances the prediction accuracy and improves the interpretation of the resulting statistical model. 

For factor models of time series, \cite{Tsai-Tsay-2010} and \cite{tsai2016doubly} impose constraints, constructed by some empirical procedures, that incorporate the inherent data structure, to both the classical and approximate factor models. Their results show that the constraints are useful tools to obtain parsimonious econometric models for forecasting, to simplify the interpretations of common factors, and to reduce the dimension. Motivated by similar concerns, we consider constrained, multi-term, and partially constrained factor models for high-dimensional matrix-variate time series. Our methods differs from \cite{Tsai-Tsay-2010} in several aspects.  First, we deal with matrix factor model and thus have the flexibility to impose row and column constraints. The interaction between the row and column constraints are explored. Second, we adopt a different set of assumptions for factor model. The factor models in \cite{Tsai-Tsay-2010} and \cite{tsai2016doubly}, following the definition in \cite{Bai-2003, Bai-Ng-2002, Bai-Ng-2007, Forni-Hallin-Lippi-Reichline-2000, forni2004generalized}, attempt to separate the common factors that affect the dynamics of most original component series from the idiosyncratic series that at most affect the dynamics of a few original time series. Such a definition is appealing in analyzing economic and financial phenomena. But the fact that idiosyncratic part may exhibit serial correlations poses technical difficulties in both identification and inference. These factor models are only asymptotically identifiable because a rigorous definition of the common factors can only be established when the dimension of time series goes to infinity. In our setting, the matrix-variate time series is decomposed into two parts: a dynamic part driven by a lower-dimensional factor time series and a static part consisting of matrix white noises. Since the white-noise series exhibits no dynamic correlations, the decomposition is unique in the sense that both the dimension of the factor process and the factor loading space are identifiable for a given finite sample size. See \cite{Lam-Yao-2011, Lam-Yao-2012, chang2015high, wang2018factor} for more detailed comparisons between these two different model definitions. 

The rest of the paper is organized as follows. Section \ref{sec:model} introduces the constrained, multi-term, and partially constrained matrix-variate factor models. Section \ref{sec:estimation} presents estimation procedures for constrained and partially constrained factor models with different constraints. Section \ref{sec:theory} investigates theoretical properties of the estimators. Section \ref{sec:simulation} presents some simulation results whereas Section \ref{sec:application} contains three applications. Section \ref{sec:summary} concludes. All proofs are in the Appendix.

\section{The Constrained Matrix Factor Model}
\label{sec:model}

For consistency in notation, we adopt the following conventions. A bold capital letter $\bA$ represents a matrix, a bold lower letter $\ba$ represents a column vector, and a lower letter $a$ represents a scalar. The $j$-th column vector and the $k$-th row vector of the matrix $\bA$ are denoted by 
$A_{\cdot j}$ and $A_{k \cdot}$, respectively.

Let $ \{ \bY_t \}_{t=1}^T $ be a matrix-variate time series, where $\bY_t$ is a $p_1 \times p_2$ matrix, that is  
\begin{equation*}
\bY_t = \left( Y_{\cdot 1, t}, \cdots, Y_{\cdot p_2, t} \right) 
= \left( \begin{array}{c} 
Y'_{1 \cdot,t} \\
\vdots \\
Y'_{p_1 \cdot,t}
\end{array}\right)
=  \left( \begin{array}{ccc} 
y_{11,t} & \cdots & y_{1p_2,t} \\
\vdots    & \ddots & \vdots \\
y_{p_11,t} & \cdots & y_{p_1p_2,t}
\end{array}\right).
\end{equation*}
\cite{wang2018factor} propose the following factor model for $\bY_t$, 
\begin{equation}  \label{eqn:mfm}
\bY_t = \bLambda \bF_t \bGamma' + \bU_t, \qquad t = 1, 2, \ldots, T,  
\end{equation}
where $\bF_t$ is a $k_1 \times k_2$ latent matrix-variate time series of common fundamental factors, $\bLambda$ is a $p_1 \times k_1$ row 
loading matrix, $\bGamma$ is a $p_2 \times k_2$ column loading matrix, and $\bU_t$ is a $p_1 \times p_2$ matrix of random errors. In Equation (\ref{eqn:mfm}), 
$(\bLambda,\bGamma)$ and $(c\bLambda,\bGamma/c)$ are equivalent if $c \neq 0$.

In Model (\ref{eqn:mfm}), we assume that $vec(\bU_t) \sim WN(\mathbf{0}, \bSigma_e)$ and is independent of the factor process $vec(\bF_t)$. That is, $\{\bU_t\}_{t=1}^T$ is a white noise matrix-variate time series and the common fundamental factors $\bF_t$ drive all dynamics and co-movement of $\bY_t$. $\bLambda$ and $\bGamma$ reflect the importance of common factors and their interactions. \cite{wang2018factor} provide several interpretations of the loading matrices $\bLambda$ and $\bGamma$. Essentially, $\bLambda$ ($\bGamma$) can be viewed as the row (column) loading matrix that reflects how each row (column) in $\bY_t$ depends on the factor matrix $\bF_t$. The interaction between the row and column is introduced through the multiplication of these terms. 

The definition of common factors in Model (\ref{eqn:mfm}) is similar to that of \cite{Lam-Yao-2011}. This decomposition facilitates model identification in finite samples and simplifies the procedure of model identification and statistical inference. However, under the definition, both the ``common factors'' defined in the traditional factor models and the serially correlated idiosyncratic components will be identified as factors. The method in \cite{wang2018factor} can only identify ``common" factors in the sense that those identifiable factors must be of certain strength. Weak factors will be left %``wrongly" 
``erroneously'' in the noise in application. %This poses challenges to the interpretation of the estimated factors, which are usually of special interest in many applications. 
Moreover, when the dimensions $p_1$ and $p_2$ are sufficiently large, interpretation of the estimated common factors $\widehat{\bF}_t$ becomes difficult because of the uncertainty and dependence involved in the estimates of the loading matrices $\bLambda$ and $\bGamma$. 

To mitigate the aforementioned difficulties  and, more importantly, to incorporate natural and known constraints among the column and row variables, we consider the following  constrained and partially constrained matrix factor models. 

A {\em constrained matrix factor model} can be written as 
\begin{equation}  \label{eqn:cmfm}
\bY_t = \bH_R \bR \bF_t \bC' \bH_C' + \bU_t,   
\end{equation} 
where $\bH_R$ and $\bH_C$ are pre-specified full column-rank $p_1 \times m_1$ and $p_2 \times m_2$ constraint matrices, respectively, and $\bR$ and $\bC$ are $m_1 \times k_1$ row loading matrix and $m_2 \times k_2$ column loading matrix, respectively. For meaningful 
constraints, we assume $k_1 \le m_1 << p_1$ and $k_2 \le m_2 << p_2$. Compared with the matrix factor model in (\ref{eqn:mfm}), we set $\bLambda = \bH_R \bR$ and $\bGamma=\bH_C \bC$ with $\bH_R$ and $\bH_C$ given. The number of parameters in the left loading matrix $\bR$ is $m_1k_1$, smaller than $p_1k_1$ of the unconstrained model. The number of parameters in the column loading matrix $\bC$ also decreases from $p_2k_2$ to $m_2k_2$. The constraint matrices $\bH_R$ and $\bH_C$ are constructed based on prior or domain knowledge of the variables. 
%For example, if $\bH_R$ consists of orthogonal binary vectors, it represents a classification or grouping of the rows of the observed matrix. 

\subsection{Examples of Constraint Matrices}  \label{sec:constraints}
We first consider discrete \textbf{covariate-induced constraint matrices}, using dummy variables. 
Continuous covariate may be segmented into regimes.
As an illustration we consider the following toy example of corporate financial matrix-valued time series. Suppose we have $8$ companies, which can be grouped according to their industrial classification (Tech and Retail) and also their market capitalization 
(Large and Medium). The two groups form $2 \times 2$ combinations as shown in Table \ref{table:company_category},

\begin{table}[ht!]
	\centering
	\makebox[0pt][c]{\parbox{0.9\textwidth}{%
			\begin{minipage}[b]{0.5\hsize} \centering
				\begin{adjustbox}{center, width=0.9\columnwidth}
				\begin{tabular}{clcc}
					& \multicolumn{1}{c}{} & \multicolumn{2}{c}{Market Cap} \\
					& \multicolumn{1}{c}{} & C1. Large & C2. Medium \\ \cline{3-4} 
					\multirow{2}{*}{Industry} & \multicolumn{1}{l|}{I1. Tech} & \multicolumn{1}{c|}{Apple, Microsoft} & \multicolumn{1}{c|}{Brocade, FireEye} \\ \cline{3-4} 
					& \multicolumn{1}{l|}{I2. Retail} & \multicolumn{1}{c|}{Walmart, Target} & \multicolumn{1}{c|}{JC Penny, Kohl's} \\ \cline{3-4} 
				\end{tabular}
			    \end{adjustbox}
			\end{minipage}		    
			\begin{minipage}[b]{0.5\hsize} \centering
				\begin{adjustbox}{center, width=0.9\columnwidth}
				\begin{tabular}{ccc}
					& Industry & Market Cap \\ \cline{2-3} 
					\multicolumn{1}{c|}{Apple} & \multicolumn{1}{c|}{I1} & \multicolumn{1}{c|}{C1} \\ 
					\multicolumn{1}{c|}{Microsoft} & \multicolumn{1}{c|}{I1} & \multicolumn{1}{c|}{C1} \\
					\multicolumn{1}{c|}{Brocade} & \multicolumn{1}{c|}{I1} & \multicolumn{1}{c|}{C2} \\  
					\multicolumn{1}{c|}{FireEye} & \multicolumn{1}{c|}{I1} & \multicolumn{1}{c|}{C2} \\ \cline{2-3} 
					\multicolumn{1}{c|}{Walmart} & \multicolumn{1}{c|}{I2} & \multicolumn{1}{c|}{C1} \\ 
					\multicolumn{1}{c|}{Target} & \multicolumn{1}{c|}{I2} & \multicolumn{1}{c|}{C1} \\ 
					\multicolumn{1}{c|}{JC Penny} & \multicolumn{1}{c|}{I2} & \multicolumn{1}{c|}{C2} \\ 
					\multicolumn{1}{c|}{Kohl's} & \multicolumn{1}{c|}{I2} & \multicolumn{1}{c|}{C2} \\ \cline{2-3} 
				\end{tabular}
			    \end{adjustbox}
			\end{minipage}			
	}}
\caption{Groups of companies by industry and market capitalization}
\label{table:company_category} 
\end{table} 

Constraint matrix $\bH_R^{(1)}$ in Table \ref{table:constraint-illus} utilizes only industrial classification. To combine both industrial classification and market cap information, we first consider an additive model constraint on the $8 \times k_1$ ($k_1 \le 3$) loading matrix $\bLambda$ in Model (\ref{eqn:mfm}). The additive model constraint means that the $i$-th row of $\bLambda$, that is, the loadings of $k_1$ row factors on the $i$-th variable, must assume the form $\blam_{i \,\cdot} = \bu_{j \, \cdot} + \bv_{l \, \cdot}$, where the $i$-th variable falls in group $({Industry}_j, {MarketCap}_l)$, $k_1$-dimensional vectors $\bu_{j \, \cdot}$ and $\bv_{l \, \cdot}$ are the loadings of $k_1$ row factors on the $j$-th market cap group and $l$-th industrial group, respectively. The most obvious way to express the additive model constraint is to use row constraints $\bH^{(2)}_R$ in Table \ref{table:constraint-illus}. Then, in the constrained matrix factor model (\ref{eqn:cmfm}), $\bH_R=\bH^{(2)}_R$ and $\bR = (\bu_{1 \, \cdot}, \bu_{2 \, \cdot}, \bv_{1 \, \cdot}, \bv_{2 \, \cdot})'$. 
 
\begin{table}[ht!]
	\centering
	\makebox[0pt][c]{\parbox{\textwidth}{%			
			\begin{minipage}[b]{0.23\hsize}\centering
				\begin{tabular}{c|cc|}
					\cline{2-3}
					& 1 & 0 \\
					& 1 & 0 \\
					& 1 & 0 \\
					\multirow{2}{*}{$\bH^{(1)}_R$=} & 1 & 0 \\ \cline{2-3} 
					& 0 & 1 \\
					& 0 & 1 \\
					& 0 & 1 \\
					& 0 & 1  \\ \cline{2-3} 
				\end{tabular}
			\end{minipage}
			\hfill
			\begin{minipage}[b]{0.32\hsize}\centering
				\begin{tabular}{c|cc|cc|}
					\cline{2-5}
					& 1 & 0 & 1 & 0 \\
					& 1 & 0 & 1 & 0 \\
					& 1 & 0 & 0 & 1 \\
					\multirow{2}{*}{$\bH^{(2)}_R$=} & 1 & 0 & 0 & 1 \\ \cline{2-5} 
					& 0 & 1 & 1 & 0 \\
					& 0 & 1 & 1 & 0 \\
					& 0 & 1 & 0 & 1 \\
					& 0 & 1 & 0 & 1 \\ \cline{2-5} 
				\end{tabular}
			\end{minipage}%
		    \hfill
		    \begin{minipage}[b]{0.4\hsize}\centering
		    	\begin{tabular}{c|cc|cc|c|}
		    		\cline{2-6}
		    		& 1 & 0 & 1 & 0 & 1 \\
		    		& 1 & 0 & 1 & 0 & 1 \\
		    		& 1 & 0 & 0 & 1 & -1 \\
		    		\multirow{2}{*}{$\bH^{(3)}_R=$} & 1 & 0 & 0 & 1 & -1 \\ \cline{2-6} 
		    		& 0 & 1 & 1 & 0 & -1 \\
		    		& 0 & 1 & 1 & 0 & -1 \\
		    		& 0 & 1 & 0 & 1 & 1 \\
		    		& 0 & 1 & 0 & 1 & 1 \\ \cline{2-6} 
		    	\end{tabular}
		    \end{minipage}%
	}}
    \caption{Illustration of constraint matrices constructed from grouping information by additive model.}
    \label{table:constraint-illus}
\end{table}

Further, we consider the constraint incorporating an interaction term between industry and market cap grouping information. Now the $i$-th row of $\bLambda$ has the form $\blam_{i \,\cdot} = \bu_{j \, \cdot} + \bv_{l \, \cdot} + \alpha_{j,l} \bw$, where $\bw$ is the $k_1$-dimensional interaction vector containing loadings of $k_1$ row factors and $\alpha_{ij}$ is the interaction term determined by $\bu_{j \, \cdot}$ and $\bv_{l \, \cdot}$ jointly. For example, 
\[
\alpha_{j,l} = \begin{cases}
1 & if \quad j=l=1 \text{ or } 2,\\
-1 & if \quad j=1, l=2 \text{ or vice versa}.
\end{cases}
\]

In this case,  for the constrained matrix factor model (\ref{eqn:cmfm}), $\bH_R=\bH^{(3)}_R$ and $\bR = (\bu_{1 \, \cdot}, \bu_{2 \, \cdot}, \bv_{1 \, \cdot}, \bv_{2 \, \cdot}, \bw)'$. Note that $\bH^{(2)}_R$ and $\bH^{(3)}_R$ here are not full column rank and can be reduced to a full column rank matrix satisfying the requirement in Section \ref{sec:estimation}. But the presentations of $\bH^{(2)}_R$ and $\bH^{(3)}_R$ are sufficient to illustrate the ideas of constructing complex constraint matrices.

To illustrate a \textbf{theory-induced constraint matrix}, we consider the yield curve latent factor model. \cite{nelson1987parsimonious} propose the Nelson-Siegel representation of the yield curve using a variation of the three-component exponential approximation to the cross-section of yields at any moment in time, 
\[
y(\tau) = \beta_1 + \beta_2 \left( \frac{1-e^{-\lambda \tau}}{\lambda \tau} \right) + \beta_3 \left( \frac{1-e^{-\lambda \tau}}{\lambda \tau} - e^{-\lambda \tau} \right),
\]
where $y(\tau)$ denotes the set of zero-coupon yields and $\tau$ denotes time to maturity. 

\cite{diebold2006forecasting} and \cite{diebold2006macroeconomy} interpret the Nelson-Siegel representation as a dynamic latent factor model where $\beta_1$, $\beta_2$, and $\beta_3$ are time-varying latent factors that capture the level (L), slope (S), and curvature (C) of the yield curve at each period $t$, while the terms that multiply the factors are respective factor loadings, that is 
\[
y(\tau) = L_t + S_t \left( \frac{1-e^{-\lambda \tau}}{\lambda \tau} \right) + C_t \left( \frac{1-e^{-\lambda \tau}}{\lambda \tau} - e^{-\lambda \tau} \right).
\]
The factor $L_t$ may be interpreted as the overall level of the yield curve since its loading is equal for all maturities. The factor $S_t$, representing the slope of the yield curve, has a maximum loading (equal to $1$) at the shortest maturity and then monotonically decays through 0 (to -1) as maturities increase. And the factor $C_t$ has a loading that is $-1$ at the shortest maturity, increases to an intermediate maturity (equal to 2) and then falls back to $-1$ as maturities increase. Hence, $S_t$ and $C_t$ capture the short-end and medium-term latent components of the yield curve. The coefficient $\lambda$ controls the rate of decay of the loading of $C_t$ and the maturity where $S_t$ has maximum loading. 

Multinational yield curve can be represented as a matrix time series $\{\bY_t\}_{t=1}^T$, 
where rows of $\bY_t$ represent time to maturity and columns of $\bY_t$ denotes countries. To capture the characteristics of loading matrix specific to the level, slope, and curvature factors, we could set row loading constraint matrix to, for example, $\bH_R = [\bh_1, \bh_2, \bh_3]$, where $\bh_1=(1,1,1,1,1)'$, $\bh_2=(1,1,0,-1,-1)'$ and $\bh_3=(-1,0,2,0,-1)$. In Section \ref{sec:simulation}, we try to mimic multinational yield curve and generate our samples from this type of constraints.

\subsection{Multi-term and partially constrained matrix factor models} \label{sec:multi-term}

If there are two ``distinct'' sets of constraints and the factors corresponding to these two sets do not interact, Model (\ref{eqn:cmfm}) can be extended to a {\em multi-term matrix factor model} as  
\begin{equation}  \label{eqn:cmfm-m}
\bY_t = \bH_{R_1} \bR_1 \bF_{1t} \bC_1' \bH_{C_1}' + \bH_{R_2} \bR_2 \bF_{2t} \bC_2' \bH_{C_2}' + \bU_t.   
\end{equation} 
For example, countries can be grouped according to their geographic locations, such as European and Asian countries, and also grouped according to their economic characteristics, such as natural resource based and manufacture based economies, and the corresponding factors may not interact with each other. 
% \textcolor{red}{To ensure identifiability, we assume $\bH_{R_2} \in \bH_{R_1}^{\perp}$ and $\bH_{C_2} \in \bH_{C_1}^{\perp}$. If not, we define $\bH_{R_2}^* = (1-\bH_{R_1} \bH'_{R_1}) \bH_{R_2}$ being the full rank version of ....or full column rank?} 

Note that (\ref{eqn:cmfm-m}) can be rewritten as (\ref{eqn:cmfm}), with $\bH_R=\begin{bmatrix} \bH_{R_1} & \bH_{R_2} \end{bmatrix}$, $\bH_C=\begin{bmatrix} \bH_{C_1} & \bH_{C_2} \end{bmatrix}$,
\[
\bR= \begin{bmatrix}
\bR_1 & 0 \\
0 & \bR_2 
\end{bmatrix},
\bC= \begin{bmatrix}
\bC_1 & 0 \\
0 & \bC_2 
\end{bmatrix},  \text{ and } 
\bF_t=\begin{bmatrix} 
\bF_{1t} & 0 \\
0 & \bF_{2t}
\end{bmatrix}.
\]
Hence (\ref{eqn:cmfm-m}) is a special case of (\ref{eqn:cmfm}) with the strong assumption that the factor matrix is block diagonal. Such a simplification can greatly enhance the interpretation of the model. 

\begin{remark}
The pre-specified constraint matrices $\bH_{R_1}$ and $\bH_{R_2}$ do not have to be orthogonal. Neither does the pair $\bH_{C_1}$ and $\bH_{C_2}$. An estimation procedure is presented in Remark \ref{remark:est_multi_nonorhtogonal} in Section \ref{sec:estimation:subsec:multiterm_const}. 
The rates of convergence will change as a result of information loss from the estimation procedure to deal with the nonorthogonality of $\bH_{R_1}$ and $\bH_{R_2}$. 
Since we can always transform non-orthogonal constraint matrices to some orthogonal constraint matrices, we shall  focus on the case when $\bH_{R_1}$ and $\bH_{R_2}$ (or $\bH_{C_1}$ and $\bH_{C_2}$) are orthogonal.
\end{remark}

In many applications, prior or domain knowledge may not be sufficiently comprehensive or may only provide a partial specification of the constraint matrices. 
In the above example, it is possible that the countries within a group react to one set of factors the same way, but differently to another set of factors. 
In such cases, a partially constrained factor model would be more appropriate. 
Specifically, a {\em partially constrained matrix factor model} can be written as 
\begin{equation}  \label{eqn:pcmfm}
\bY_t = \begin{bmatrix} \bH_{R_1} \bR_1 & \bLambda_2 \end{bmatrix} 
\begin{bmatrix}
\bF_{11,t} & \bF_{12,t} \\
\bF_{21,t} & \bF_{22,t}
\end{bmatrix}
\begin{bmatrix} \bC'_1 \bH'_{C_1} \\ \bGamma_2' \end{bmatrix}
+ \bU_t,   
\end{equation}   
where $\bH_{R_1}$, $\bR_1$, $\bH_{C_1}$ and $\bC_1$ are defined similarly  as those in (\ref{eqn:cmfm-m}). 
$\bF_{ij,t}$'s are common matrix factors corresponding to the interactions of the row and column loading space spanned by the columns of $\bH_R$ and $\bH_C$ and their complements, $\bLambda_2$ is $p_1 \times q_1$ row loading matrix and $\bGamma_2$ is a $p_2 \times q_2$ column loading matrix. 
Again, we have $q_1 < p_1$, $q_2 < p_2$ and $vec(\bF_{ij,t})$'s are independent with $vec(\bU_t)$. 
We assume that $\bH'_{R_1} \bLambda_2 = \utwi{0}$ and $\bH'_{C_1} \bGamma_2 = \utwi{0}$, because all the row loadings that are in the space of $\bH_{R_1}$ and all the column loadings that are in the space of $\bH_{C_1}$  could be absorbed into the first parts of loading matrices. 
Thus, we could explicitly rewrite the model as 
\begin{equation}  \label{eqn:pcmfm_explicit}
\bY_t = \begin{bmatrix} \bH_{R_1} \bR_1 & \bH_{R_2} \bR_2 \end{bmatrix} 
\begin{bmatrix}
\bF_{11,t} & \bF_{12,t} \\
\bF_{21,t} & \bF_{22,t}
\end{bmatrix}
\begin{bmatrix} \bC'_1 \bH'_{C_1} \\ \bC'_2 \bH'_{C_2} \end{bmatrix}
+ \bU_t,  
\end{equation}   
where $\bH_{R_2}$ is a $p_1 \times (p_1 - m_1)$ constraint matrix satisfying $\bH_{R_1}' \bH_{R_2} = \utwi{0}$, $\bH_{C_2}$ is a $p_2 \times (p_2-m_2)$ constraint matrix satisfying $\bH_{C_1}' \bH_{C_2} =\utwi{0}$, $\bR_2$ is $(p_1 - m_1) \times q_1$ row loading matrix, and $\bC_2$ is a $(p_2-m_2) \times q_2$ column loading matrix. 

In the special case when $\bF_{21,t} = \mathbf{0}$ and $\bF_{12,t} = \mathbf{0}$, Model (\ref{eqn:pcmfm}) can be further simplified as 
\begin{equation}  \label{eqn:pcmfm_explicit_simple}
\bY_t = \bH_{R_1} \bR_1 \bF_{11,t} \bC'_1 \bH'_{C_1} +  \bH_{R_2} \bR_2 \bF_{22,t} \bC'_2 \bH'_{C_2} + \bU_t.  
\end{equation} 

Model (\ref{eqn:pcmfm_explicit_simple}) is different from the multi-term model of (\ref{eqn:cmfm-m}) in that the matrix $\bH_{R_2}$ in (\ref{eqn:pcmfm_explicit}) is induced from $\bH_{R_1}$ while the $\bH_{R_2}$ in (\ref{eqn:cmfm-m}) is an informative constraint, with a lower dimension.

In the special case when $\bH_{C_1}=\bI_{p_1}$ (there is no column constraint), Model (\ref{eqn:pcmfm_explicit}) becomes
\[
\bY_t = \begin{bmatrix} \bH_{R_1} \bR_1 & \bH_{R_2} \bR_2 \end{bmatrix} 
\begin{bmatrix}
\bF_{1,t} \\
\bF_{2,t}
\end{bmatrix}
\bC'
+ \bU_t,
\]
where $\bF_{1,t} = [\bF_{11,t}, \bF_{12,t}]$ and $\bF_{2,t} = [\bF_{21,t}, \bF_{22,t}]$. 
The left loading matrix still spans the entire $p_1$ dimensional space, but the first part of loading matrix $\bR_1$ has a clearer interpretation.

The partially constrained matrix factor model (\ref{eqn:pcmfm_explicit}) incorporates partial information $\bH_{R_1}$ and $\bH_{C_1}$ in the unconstrained model (\ref{eqn:mfm}) without ignoring the possible remainders. 
If we include all four matrix factors in the four subspaces divided by the interactions of $\bH_{R_1}$ and $\bH_{C_1}$ and their complements, the number of parameters in (\ref{eqn:pcmfm_explicit}) is the same as that in the unconstrained model (\ref{eqn:mfm}). 
However, as shown by Theorem \ref{theorem:conv_rate_loading_mat} in Section \ref{sec:theory}, the rates of convergence are faster than those of the unconstrained matrix factor model. 
Furthermore, in many applications, inclusion of only two matrix-factor terms is adequate in explaining a high percentage of variability, as exemplified by the three applications in Section \ref{sec:application}.  

\begin{remark}
Subpanel structure in multivariate time series is encountered frequently in real applications. For example, macroeconometric data often consist of large panels of time series which can be further divided into smaller but still quite large subpanels or blocks. Built upon \cite{forni2004generalized} and  \cite{hallin2007determining}, \cite{hallin2011dynamic} considered $n$-dimensional random variable $\bx = [\by' \; \bz']'$ with sub-panel vectors $\by \in \mathbb{R}^{n_y}$ and $\bz \in \mathbb{R}^{n_z}$ and proposed a method to identify and estimate joint and block-specific common factors. 
There are connections between the subpanel structure and the constrained structure considered in this paper. Both approaches produce certain block structures in the loading matrix.
Consider the vector factor model case. With two subpanels, the model becomes
\begin{equation*}%\label{equ:vector-factor}
\begin{bmatrix}\by \\ \bz\end{bmatrix} = \begin{bmatrix} \bA_{11} & \bA_{12} & \bzero \\ \bA_{21} & \bzero & \bA_{23} \end{bmatrix} \begin{bmatrix} \bff_1 \\ \bff_2 \\ \bff_3 \end{bmatrix} + \begin{bmatrix} \bepsilon_y \\ \bepsilon_z
\end{bmatrix}.
\end{equation*}
Such a model can be constructed under the constraint approach by specifying
\[
\bH=\begin{bmatrix}
\bI & \bI & \bzero \\
\bI & \bzero & \bI \\
\end{bmatrix}
\mbox{\ \ and \ \ }
\bR=\begin{bmatrix}\bA_{11} & \bzero & \bzero \\
\bzero & \bA_{12} & \bzero \\
\bA_{21}-\bA_{11} & \bzero & \bA_{23} \\
\end{bmatrix},
\]
where $\bI$'s and $\bzero$'s are identity and zero matrices of proper dimensions.
However,
our current estimation procedure is not able to force certain submatrice
in $\bR$ to zero, though the model can be turned into a
multi-term factor model as discussed in Section \ref{sec:multi-term}.
On the other hand, the constraint approach is more flexible in introducing
various types of structure in the loading matrix as illustrated in Section \ref{sec:constraints}.
\end{remark}

The benefits of considering partially constrained matrix factor models are two-folds. Firstly, the model is capable of 
identifying, from the complement spaces of $\bH_R$ and $\bH_C$, the factors that are unknown to researchers. In this case, the dimensions of $\bF_{22,t}$ are typically much smaller than those of $\bF_{11,t}$ even though the loading matrices $\bR_2$ and $\bC_2$ still have large numbers 
of rows $(p_1-m_1)$ and $(p_2-m_2)$, respectively. This is because the constraint part should have accommodated the main and key common factors. The spirit is similar to the two-step estimation of  \cite{Lam-Yao-2012} in which one fits a second-stage factor model to the residuals obtained by  subtracting the common part of the first-stage factor model. 

The second benefit is that the model is able to identify the factors corresponding to the pre-specified constraint matrices and their inherit interpretation. That is, $\bF_{11,t}$ represents the factor matrix with row and column factors affecting the observed matrix-variate time series in the way as specified by the constraints $\bH_R$ and $\bH_C$ completely. Consider the multinational macroeconomic index example. If $\bH_R$ is built from the country classification information, how the rows in $\bF_{11,t}$ affect the observations can be completely explained by the country groups instead of individual countries and the row factors in $\bF_{11,t}$ have a clearer interpretation related to the classification.  In many practical applications, researchers are interested in obtaining specific latent factors related to some domain theories and use these specific factors to predict future values of interest as guided by domain theories. For example, in 
the yield curve example in  
Section~\ref{sec:constraints}, economic theory implies that the level, slope, and curvature factors affect the observations in the way specified by, for example,    $\bH_R = [\bh_1, \bh_2, \bh_3]$, where $\bh_1=(1,1,1,1,1)'$, $\bh_2=(1,1,0,-1,-1)'$, and $\bh_3=(-1,0,2,0,-1)$. Then the estimation method in Section \ref{sec:estimation} is capable of isolating $\bH_{R_1} \bR_1 \bF_{11,t} \bC_1' \bH'_{C_1}$, hence correctly estimating the loadings and the specified level, slope, and curvature factors in the constrained spaces. As a result, the constrained factor model can serve as a method to identify and isolate specific factors suggested by domain theories or prior knowledge. 

\section{Estimation Procedure}
\label{sec:estimation}

Similar to all factor models, identification issue exits in the constrained matrix-variate factor model (\ref{eqn:cmfm}). Let $\bO_1$ and $\bO_2$ be two invertible matrices of size $k_1 \times k_1$ and $k_2 \times k_2$. Then the triples $(\bR, \bF_t, \bC)$ and $(\bR \bO_1, \bO_1^{-1} \bF_t \bO_2^{-1}, \bO_2 \bC)$ are equivalent under Model (\ref{eqn:cmfm}). Here, we may assume that the columns of $\bR$ and $\bC$ are orthonormal, that is, $\bR'\bR=\bI_{k_1}$ and $\bC'\bC=\bI_{k_2}$, where $\bI_d$ denotes the $d\times d$ identity matrix. Even with these constraints, $\bR$, $\bF_t$ and $\bC$ are not uniquely determined in (\ref{eqn:cmfm}), as aforementioned replacement is still valid for any orthonormal $\bO$. However, the column spaces of the loading matrices $\bR$ and $\bC$ are uniquely determined. Hence, in the following sections, we focus on the estimation of the column spaces of $\bR$ and $\bC$. We denote the row and column factor loading spaces by $\mathcal{M}(\bR)$ and $\mathcal{M}(\bC)$, respectively. For simplicity, we suppress the matrix column space notation and use the matrix notation directly.

\subsection{Orthogonal Constraints} 
\label{sec:estimation:subsec:orth_const}

We start with the estimation of the constrained matrix-variate factor model (\ref{eqn:cmfm}). The approach follows the ideas of \cite{Tsai-Tsay-2010} and \cite{wang2018factor}. In what follows, we illustrate the estimation procedure for the column space of $\bR$. The column space of $\bC$ can be obtained similarly from the transpose of $\bY_t$'s. For ease in representation, we assume that the process $\bF_t$ has mean $\bzero$, and the observation $\bY_t$'s are centered and standardized throughout the paper. 

Suppose we have orthogonal constraints $\bH'_R \bH_R = \bI_{m_1}$ and $\bH'_C \bH_C = \bI_{m_2}$. Define the transformation $\bX_t = \bH'_R \bY_t \bH_C$. It follows from (\ref{eqn:cmfm}) that 
\begin{equation}  \label{eqn:cmfm_trans_orth}
\bX_t = \bR \bF_t \bC' + \bE_t, \qquad t = 1, 2, \ldots, T,  
\end{equation}
where $\bE_t = \bH'_R \bU_t \bH_C$. 

This transformation projects the observed matrix time series into the constrained space. For example, if $\bH_R$ is the orthonormal matrix corresponding to the group constraint %in (\ref{hr}), 
$\bH^{(1)}_R$ of Table~\ref{table:constraint-illus}, then $\bH_R^{'}\bY_t$ is a $2\times p_2$ matrix, with the first row being the normalized average of the rows of $\bY_t$ in the first group and the second row being that in the second group. Such an operation conveniently incorporates the constraints while reduces the dimension of data matrix from $p_1\times p_2$ to $m_1\times m_2$, making the analysis more efficient. 

Since $\bE_t$ remains to be a white noise process, the estimation method in \cite{wang2018factor} directly applies to the transformed $m_1 \times m_2$ matrix time series  $\bX_t$ in Model (\ref{eqn:cmfm_trans_orth}). For completeness, we outline briefly the procedure. See \cite{wang2018factor} for details. 

To facilitate the estimation, we use the QR decomposition $\bR=\bQ_1 \bW_1$ and 
$\bC=\bQ_2 \bW_2$. 
The estimation of column spaces of $\bR$ and $\bC$ is equivalent to the estimation of column spaces of $\bQ_1$ and $\bQ_2$. 
Thus, Model (\ref{eqn:cmfm_trans_orth}) can be re-expressed as 
\begin{equation}  \label{eqn:cmfm_trans_orth_qr}
\bX_t = \bR \bF_t \bC' + \bE_t = \bQ_1 \bZ_t \bQ_2' + \bE_t, \qquad t = 1, 2, \ldots, T,  
\end{equation}
where $\bZ_t = \bW_1 \bF_t \bW_2'$, $\bQ_1' \bQ_1=\bI_{m_1}$, 
and $\bQ_2' \bQ_2=\bI_{m_2}$.

Let $h$ be a positive integer. For $i,j = 1, 2, \ldots, m_2$, define 
\begin{align}  
& \bOmega_{zq, ij}(h) = \frac{1}{T-h} \sum_{t=1}^{T-h} Cov(\bZ_t Q_{2, i \cdot}, \bZ_{t+h} Q_{2, j \cdot})     \label{eqn:Omega_zqijh}, \text{ and} \\
& \bOmega_{x, ij}(h) = \frac{1}{T-h} \sum_{t=1}^{T-h} Cov(X_{t, \cdot i}, X_{t+h, \cdot j}), \label{eqn:Omega_xijh}
\end{align}
which can be interpreted as the auto-cross-covariance matrices at lag $h$ between column $i$ and column $j$ of $\{\bZ_t \bQ'_{2}\}_{t=1}^T$ and $\{\bX_t\}_{t=1}^T$, respectively. 
For $h>0$, $\bOmega_{x, ij}(h)$ defined in (\ref{eqn:Omega_xijh}) does not involve the covariance terms incurred by $\bE_t$ because of the whiteness condition. 

For a fixed $h_0 \ge 1$ satisfying Condition 2 in Appendix A, %\ref{appendix:proofs}
define
\begin{equation}  \label{eqn:M_def}
\bM = \sum_{h=1}^{h_0} \sum_{i=1}^{m_2} \sum_{j=1}^{m_2} \bOmega_{x,ij}(h) \bOmega_{x,ij}(h)' = \bQ_1 \left\{ \sum_{h=1}^{h_0} \sum_{i=1}^{m_2} \sum_{j=1}^{m_2} \bOmega_{zq, ij}(h) \bOmega_{zq, ij}(h)' \right\} \bQ_1'.  
\end{equation}

Since $\bM$ and the matrix sandwiched by $\bQ_1$ and $\bQ'_1$ are positive definite matrices, Equation (\ref{eqn:M_def}) implies that the eigen-space of $\bM$ is the same as the column space of $\bQ_1$ if the middle term is full rank (Condition~2 in Appendix~A. %\ref{appendix:proofs})
Hence, $\mathcal{M}(\bQ_1)$ can be estimated by the space spanned by the eigenvectors of the sample version of $\bM$. 
%The columns of $\bQ_1$ can be obtained as the $k_1$ orthogonal eigenvectors of the matrix $\bM$ corresponding to its $k_1$ orthogonal eigenvectors of matrix $\bM$ corresponding to its $k_1$ non-zero eigenvalues arranged in the descending order. 
The normalized eigenvectors $\bq_{1}, \ldots, \bq_{k_1}$ corresponding to the $k_1$ nonzero eigenvalues of $\bM$ are uniquely defined up to a sign change. Thus $\bQ_1$ is uniquely defined by $\bQ_1 = (\bq_{1}, \ldots, \bq_{k_1})$ up to a sign change. We estimate $\widehat{\bQ}_1= (\widehat{\bq}_{1}, \ldots, \widehat{\bq}_{k_1})$ as a representative of $\mathcal{M}(\bQ_1)$ or $\mathcal{M}(\bR)$

The estimation procedure is based on the sample version of these quantities. For $h \ge 1$ and a prescribed positive integer $h_0$, define the sample version of $\bM$ in (\ref{eqn:M_def}) as the following
\begin{equation}  \label{eqn:Mhat_def}
\widehat{\mathbf{M}} = \sum_{h=1}^{h_0} \sum_{i=1}^{m_2} \sum_{j=1}^{m_2} \widehat{\mathbf{\Omega}}_{x,ij}(h) \widehat{\mathbf{\Omega}}_{x,ij}(h)', \text{  where} \quad \widehat{\mathbf{\Omega}}_{x, ij}(h) = \frac{1}{T-h} \sum_{t=1}^{T-h} X_{t, \cdot i} X'_{t+h, \cdot j}. 
\end{equation}

Then, $\mathcal{M}(\bQ_1)$ can be estimated by $\mathcal{M}(\widehat{\bQ}_1)$, where $\widehat{\bQ}_1=(\widehat{\bq}_1, \ldots, \widehat{\bq}_{k_1})$ and $\widehat{\bq}_i$ is an eigenvector of $\widehat{\bM}$, corresponding to its $i$-th largest eigenvalue.
The $\bQ_2$ is defined similarly for the column loading matrix $\bC$ and $\mathcal{M}(\widehat{\bQ}_2)$ and $\widehat{\bQ}_2$ can be estimated with the same procedure to the transpose of $\bX_t$. Consequently, we estimate the normalized factors and residuals, respectively, by $\widehat{\bZ}_t = \widehat{\bQ}'_1 \bX_t \widehat{\bQ}_2$ and $\widehat{\bU}_t = \bY_t - \bH_R \widehat{\bQ}_1 \widehat{\bZ}_t \widehat{\bQ}'_2 \bH'_C$.

The above estimation procedure assumes that the number of row factors $k_1$ is known. To determine $k_1$, \cite{wang2018factor} used the eigenvalue ratio-based estimator of \cite{Lam-Yao-2012}. Let $\widehat{\lambda}_1 \ge \widehat{\lambda}_2 \ge \cdots \ge \widehat{\lambda}_{m_1} \ge 0$ be the ordered eigenvalues of $\widehat{\bM}$. The ratio-based estimator for $k_1$ is defined as 
\begin{equation*}  %\label{eqn:eigen_ratio_k1}
\widehat{k}_1 = \arg \min_{1 \le j \le K} \frac{\widehat{\lambda}_{j+1}}{\widehat{\lambda}_j}, 
\end{equation*}
where $k_1 \le K \le p_1$ is an integer. In practice we may take $K = \lceil p_1/2 \rceil$. $k_2$ can be estimated with the same procedure with the $\widehat{\bM}$-matrix corresponding to the transpose of $\bX_t$.

Although the estimation procedure on the transformed series $\bX_t$ is exactly the same as that of \cite{wang2018factor}, the asymptotic properties of the estimator are different due to the transformation, as shown in Section \ref{sec:theory}, and $\bX_t$ is of lower dimension.

\subsection{Nonorthogonal Constraints}  
\label{sec:estimation:subsec:nonorth_const}

If the constraint matrix $\bH_R$ (or $\bH_C$) is not orthogonal, we can perform column orthogonalization and standardization, similar to that in \cite{Tsai-Tsay-2010}. Specifically, we obtain 
\[
\bH_R = \bTheta_R \bK_R,
\]
where $\bTheta_R$ is an orthonormal matrix and $\bK_R$ is a $m_1 \times m_1$ upper triangular matrix with nonzero diagonal elements. $\bH_C = \bTheta_C \bK_C$ can be obtained in the same way. 

Letting $\bX_t = \bTheta'_R \bY_t \bTheta_C$, $\bR^*= \bK_R \bR$, and $\bC^*=\bK_C\bC$, 
we have 
\begin{equation}  \label{eqn:cmfm_trans_nonorth}
\bX_t = \bR^* \bF_t \bC^{*'} + \bE_t, \qquad t = 1, 2, \ldots, T, 
\end{equation}
where $\bE_t = \bTheta'_R \bU_t \bTheta_C$. Since $\bE_t$ remains a white noise process, we apply the same estimation method as that in Section \ref{sec:estimation:subsec:orth_const} to obtain $\widehat{\bQ}^*_1$ and $\widehat{\bQ}^*_2$ as the representatives of $\mathcal{M}(\widehat{\bR}^*)$ and $\mathcal{M}(\widehat{\bC}^*)$. Then the estimators of $\bR$ and $\bC$ are $\widehat{\bR} = \bK_R^{-1} \widehat{\bQ}^*_1$ and $\widehat{\bC} = \bK_C^{-1} \widehat{\bQ}^*_2$. Note that $\bK_R$ and $\bK_C$ are invertible lower triangular matrices.

\subsection{Multi-term Constrained Matrix Factor Model}
\label{sec:estimation:subsec:multiterm_const}

Without loss of generality, we assume that both row and column constraint matrices are orthogonal matrices. If $\bH_{R_1}$ and $\bH_{R_2}$ (or $\bH_{C_1}$ and $\bH_{C_2}$) are orthogonal, we obtain, for $t = 1, 2, \ldots, T$,
\begin{eqnarray}  
\bH'_{R_1} \bY_t \bH_{C_1} & = & \bR_1 \bF_{1,t} \bC_1' + \bH'_{R_1} \bU_t \bH_{C_1},   \label{eqn:multi-1-tr-ortho} \nonumber \\
\bH'_{R_2} \bY_t \bH_{C_2} & = & \bR_2 \bF_{2,t} \bC_2' + \bH'_{R_2} \bU_t \bH_{C_2},  \label{eqn:multi-2-tr-ortho} \nonumber
\end{eqnarray}
where $\bH'_{R_1} \bU_t \bH_{C_1}$ and $\bH'_{R_2} \bU_t \bH_{C_2}$ are white noises. The estimators of $\widehat{\bR}_1$, $\widehat{\bC}_1$, $\widehat{\bF}_{1,t}$, $\widehat{\bR}_2$, $\widehat{\bC}_2$ and $\widehat{\bF}_{2,t}$ can be obtained by applying the estimation procedure described in Section \ref{sec:estimation:subsec:orth_const} to $\bH'_{R_1} \bY_t \bH_{C_1}$ and $\bH'_{R_2} \bY_t \bH_{C_2}$, respectively. 

\begin{remark} \label{remark:est_multi_nonorhtogonal}
For multi-term constrained model (\ref{eqn:cmfm-m}), $\bH_{R_1}$ and $\bH_{R_2}$ (or $\bH_{C_1}$ and $\bH_{C_2}$) may not necessarily be orthogonal. In this case, we illustrate the estimation procedure for the column loadings. Define projection matrices $\bP_{\bH^{\perp}_{R_1}} = \bI - \bH_{R_1} \bH'_{R_1}$ and $\bP_{\bH^{\perp}_{R_2}} = \bI - \bH_{R_2} \bH'_{R_2}$, which represent the projections onto the spaces perpendicular to the column spaces of $\bH_{R_1}$ and $\bH_{R_2}$, respectively. Left multiplying Equation (\ref{eqn:cmfm-m}) by $\bP_{\bH^{\perp}_{R_2}}$ and $\bP_{\bH^{\perp}_{R_1}}$, respectively, and taking transpose of the resulting matrices, we have 
\begin{eqnarray*}
\bY'_t \bP_{\bH^{\perp}_{R_2}} & = & \bH_{C_1} \bC_1 \bF'_{1,t} \bR'_1 \bH'_{R_1} \bP_{\bH^{\perp}_{R_2}} + \bU'_t \bP_{\bH^{\perp}_{R_2}}, \\
\bY'_t \bP_{\bH^{\perp}_{R_1}} & = & \bH_{C_2} \bC_2 \bF'_{2,t} \bR'_2 \bH'_{R_2} \bP_{\bH^{\perp}_{R_1}} + \bU'_t \bP_{\bH^{\perp}_{R_1}}, 
\end{eqnarray*}
where $ \bP_{\bH^{\perp}_{R_2}} \bU_t$ and $\bP_{\bH^{\perp}_{R_1}} \bU_t$ are white noises. The column loading estimators $\widehat{\bC}_1$ and $\widehat{\bC}_2$ can be obtained by applying the procedure described in Section \ref{sec:estimation:subsec:orth_const} to $ \bH'_{C_1} \bY'_t \bP_{\bH^{\perp}_{R_2}}$ and $\bH'_{C_2} \bY'_t \bP_{\bH^{\perp}_{R_1}}$, respectively. Note that the $p_1 \times m_1$ matrix $\bP_{\bH^{\perp}_{R_2}} \bH_{R_1}$ is no longer full rank or orthonormal. However, the row and column loading spaces and latent factors can be fully recovered if the dimension of the reduced constrained loading spaces still larger than the dimensions of the latent factor spaces. However, the rates of convergence will change. For example, the rate of convergence of $\widehat{\bC}_1$ will depend on $\norm{\bP_{\bH^{\perp}_{R_2}} \bH_{R_1} \bR_1}^2_2$ instead of $\norm{\bH_{R_1} \bR_1}^2_2$. 
\end{remark}

\subsection{Partially Constrained Matrix Factor Model}
\label{sec:estimation:subsec:partial_const}

For the partially constrained matrix factor model (\ref{eqn:pcmfm_explicit}), we assume that $\bH_{R_1}' \bH_{R_2} = \utwi{0}$ and $\bH_{C_1}' \bH_{C_2} = \utwi{0}$. Define the  transformation $\bX^{(lk)}_t = \bH'_{R_l} \bY_t \bH_{C_k}$ for $l,k= 1,2$. Then the transformed data follow the structure, 
\begin{equation}  \label{eqn:pcmfm_trans_orth}
\bX^{(lk)}_t = \bR_l \bF_{lk,t} \bC_k' + \bE^{(lk)}_t, \quad l,k = 1,2,   \nonumber
\end{equation}
where $\bE^{(lk)}_t = \bH'_{R_l} \bU_t \bH_{C_k}$ remains a white noise process. 

Let $\bM^{(lk)}$ represent the $\bM$ matrix defined in (\ref{eqn:M_def}) for each $\bX^{(lk)}_t$, $l,k=1,2$. Define $\bM^{(l\cdot)} = \sum_{k=1}^{2} \bM^{(lk)} $ for $l=1,2$, then 
\begin{equation} 
\bM^{(l\cdot)} = \bQ^{(l)}_1 \left\{ \sum_{k=1}^{2} \sum_{h=1}^{h_0} \sum_{i=1}^{m_2} \sum_{j=1}^{m_2} \bOmega^{(lk)}_{zq, ij}(h) \bOmega^{(lk)}_{zq, ij}(h)'  \right\} \bQ^{(l)'}_1 , \quad l=1,2, \label{eqn:M_partial}
\end{equation}
has the same column space as that of $\bR_l$, for $l=1, 2$, respectively. 

The estimators of $\widehat{\bR}_l$, $l=1,2$, can be obtained by applying eigen-decomposition on the sample version of $\bM^{(l\cdot)}$ defined similarly to (\ref{eqn:Mhat_def}). $\bC_k$, $k = 1,2$, can be obtained by using the same procedure on the transposes of $\bX^{(lk)}_t$ for $l,k= 1,2$. In the special case of Model (\ref{eqn:pcmfm_explicit_simple}) if $\bF_{21,t} = \mathbf{0}$ and $\bF_{12,t} = \mathbf{0}$, the above estimation is essentially the same procedure as those described in Section \ref{sec:estimation:subsec:orth_const} applying to $\bX^{(ll)}_t$ for $l= 1,2$. 

This procedure effectively projects the observed matrix time series $\bY_t$ into four orthogonal subspaces, based on the constraints obtained from the domain knowledge or some empirical procedure. Because $\bX^{(lk)}_t, l,k=1,2$ are orthogonal, they can be analyzed separately. In our setting, we divide a $p_1 \times p_1$ ambient space of row loading matrix into two orthogonal $p_1 \times m_1$ and $p_1 \times (p_1 - m_1)$ subspaces. The estimation procedure for the partially constrained model ensures the structural requirement that $\bX^{(l1)}_t$ and $\bX^{(l2)}_t$ share the same row loading matrix for the same $l$  without sacrificing the dimension reduction benefit from column space division. More generally, we could divide the space of loading matrix into more than two parts to accommodate each application. Under this partially constrained model, the orthogonality assumption between $\bF_{lk,t}, l,k=1,2$ is not important as they are latent variables.

\begin{remark}
In situations when the prior or domain knowledge captures most major factors, it is reasonable to assume that $m_i$ grows slower than $p_i$ and the row (column) factor strength (defined in Condition~6 in Section~\ref{sec:theory} ) of the main factor $\bF_{11,t}$ is no weaker than that of the remainder factor $\bF_{22,t}$. Improved estimators of $\hat{\bR}_l$, $l=1,2$, can be obtained by applying eigen-decomposition on the sample version of $\bM^{(l1)}$ defined similarly to (\ref{eqn:Mhat_def}). Improved estimators of $\hat{\bC_k}$, $k = 1,2$, can be obtained by using the same procedure on the transposes of $\bX^{(1k)}_t$ for $k= 1,2$. %Here, the estimation procedure discards the noisy part in (\ref{eqn:M_partial}) and results in improved estimators. 
\end{remark}

\section{Theoretical Properties}
\label{sec:theory}

In this section, we present the convergence rates of the estimators under the setting that $p_1$, $p_2$, $m_1$, $m_2$ and $T$ all go to infinity while the dimensions $k_1$, $k_2$ and the structure of the latent factor are fixed over time. 
In what follows, let $\norm{\bA}_2$, $\norm{\bA}_F$ and $\norm{\bA}_{min}$ denote the 
spectral norm, Frobenius norm, and the smallest nonzero singular value of $\bA$, respectively. 
%They are defined as $\norm{\bA}_1 = \max_i \sum_{j} |\bA_{ij}|$, $\norm{\bA}_2 = \sqrt{\lambda_{max}(\bA'\bA)}$,  and $\norm{\bA}_F = \sqrt{tr(\bA'\bA)}$. 
%$\norm{\bA}_{min}$ denotes the positive square root of the minimal eigenvalue of $\bA'\bA$ or $\bA\bA'$, whichever is a smaller matrix. 
When $\bA$ is a square matrix, we denote by $tr(\bA)$, $\lambda_{max}(\bA)$ and $\lambda_{min}(\bA)$ the trace, maximum and minimum eigenvalues of the matrix $\bA$, respectively. 
For two sequences $a_N$ and $b_N$, we write $a_N \asymp b_N$ if $a_N = O(b_N)$ and $b_N = O(a_N)$.

The asymptotic convergence rates are significantly different from those in 
\cite{wang2018factor} due to the constraints. The results reveal more clearly the impact 
of the constraints on signals and noises and the interaction between them. 
We only consider the case of the orthogonal constrained model (\ref{eqn:cmfm}). Asymptotic properties of nonorthogonal, multi-term, and partially constrained matrix factor model are trivial extensions. 

Several regularity conditions (Conditions 1 to 5) are listed in the Appendix. They are similar to those in \cite{wang2018factor} and are used to derive the limiting behavior of (\ref{eqn:Mhat_def}) towards its population version. 
The following condition requires some discussion. 
\vspace{-0.1in}
\paragraph*{Condition 6.} \label{cond:factor_strength}
\textbf{Factor Strength.} There exist constants $\delta_1$ and $\delta_2$ in $[0, 1]$ such that $\norm{\bH_R \bR}^2_2 \asymp p_1^{1-\delta_1} \asymp \norm{\bH_R \bR}^2_{min}$ and $\norm{\bH_C \bC}^2_2 \asymp p_2^{1-\delta_2} \asymp \norm{\bH_C \bC}^2_{min}$. 
\vspace{0.15in}

Since only $\bY_t$ is observed in Model (\ref{eqn:cmfm}), how well we can recover the factor $\bF_t$ from $\bY_t$ depends on the `factor strength' reflected by the coefficients in the row and column factor loading matrices $\bH_R \bR$ and $\bH_C \bC$. For example, in the 
case of $\bH_R \bR= \utwi{0}$ or $\bH_C \bC=\utwi{0}$, 
$\bY_t$ carries no information on $\bF_t$. In the following, 
we assume $\norm{\bF_t}$ does not change as $p_1$, $p_2$, $m_1$, and $m_2$ change. 

The rates $\delta_1$ and $\delta_2$ in \nameref{cond:factor_strength} are called the strength for the row factors and the column factors, respectively. If $\delta_1 = 0$, the corresponding row factors are called strong factors because Condition 6 implies that the factors have impacts on the majority of $p_1$ vector time series. The amount of information that observed process $\bY_t$ carries about the strong factors increases at the same rate as the number of observations or the amount of noise increases. If $\delta_1>0$, the row factors are weak, which means that the information contained in $\bY_t$ about the factors grows more slowly than the noises introduced as $p_1$ increases. The smaller the $\delta's$, the stronger the factors. In the strong factor case, the loading matrix is dense. See \cite{Lam-Yao-2011} for further discussions. 

If we restrict $\bH_R$ to be orthonormal, $||\bH_R \bR||_2^2 = ||\bR||_2^2 \asymp p_1^{1-\delta_1}$ and there is an interplay between $\bH_R$ and $\bR$ as $p_1$ increases. In order for $\bH_R$ to remain orthonormal, when $p_1$ increases, each element of $\bH_R$ decreases at the rate of $p_1^{-1/2}$. At the same time, each element of $\bR$ on average increases at the rate of $\sqrt{p_1^{1-\delta_1}/m_1}$. The column factor loading $||\bH_C \bC||_2^2$ behaves in the same way. As $p_1$ and $p_2$ increase, each element of the transformed error $\bE_t$ remains a growth rate of $1$ under Condition 3 (see Lemma~1 in Appendix~A, %\ref{appendix:proofs}), 
but the dimension of $\bE_t$ is $m_1 \times m_2$ which grows at a slower rate than $p_1 \times p_2$. The factor strength is defined in terms of the observed dimension $p_1$ and $p_2$ and the overall loading matrices $\bH_R\bR$ and $\bH_C\bC$, but clearly how $m_1$ and $m_2$ increase with $p_1, p_2$ is also important because it controls the signal-noise ratio in the constrained model. %For example, if $m_i/p_i = c_i$, $c_i < 1$, then $||\bR||_2^2 ||\bC||_2^2 \asymp m_1^{1-\delta_1}m_2^{1-\delta_2}/c_1^{1-\delta_1}c_2^{1-\delta_2}$ and if $m_i=p_i^{\alpha_i}$, $\alpha_i < 1$,  then $||\bR||_2^2 ||\bC||_2^2 \asymp m_1^{(1-\delta_1)/\alpha_1} m_2^{(1-\delta_2)/\alpha_2}$, while $||\bE_t||_2^2 \asymp m_1 m_2$ in both cases. The signal-noise ratios differ, however, in either case, the signal-noise ratio is larger than $m_1^{-\delta_1}m_2^{-\delta_2}$ -- the signal-noise ratio of a unconstrained matrix factor model. 

%\paragraph{Condition 7.} \label{cond:Ft_E{t+h}_uncorrelated}
%The covariance matrix $Cov(\bU_t, \bF_t) = 0$ for all $s \le t$. 

%\textcolor{gray}{\nameref{cond:Ft_E{t+h}_uncorrelated} relaxes independence assumption between $\{ \bF_t \}$ and $\{ \bU_t \}$ imposed in most factor model literature. It allows future factors to be correlated with past white noise. (\cite{wang2018factor} does not have this relax condition. Need to make a little change to the proofs.)}

%\subsection{Asymptotic Properties}  
%\label{sec:theory:subsec:asymptotic}

We have the following theorems for the constrained matrix factor model. Asymptotic properties for the multi-term and the partially constrained models are similar and can be derived easily. 

\begin{theorem} \label{theorem:conv_rate_loading_mat}
Under Conditions 1-6 and $m_1 p_1^{-1+\delta_1} m_2 p_2^{-1+\delta_2} T^{-1/2} = o(1)$, as $m_1$, $p_1$, $m_2$, $p_2$, and $T$ go to $\infty$, it holds that
\begin{eqnarray}
& \norm{ \widehat{\bQ}_1 - \bQ_1 }_2 & = O_p \left( \max \left( T^{-1/2}, \; \frac{m_1}{p_1^{1-\delta_1}} \frac{m_2}{p_2^{1-\delta_2}}  T^{-1/2} \right) \right)  \nonumber , \\
& \norm{ \widehat{\bQ}_2 - \bQ_2 }_2 & = O_p \left( \max \left( T^{-1/2}, \; \frac{m_1}{p_1^{1-\delta_1}} \frac{m_2}{p_2^{1-\delta_2}}  T^{-1/2} \right) \right). \nonumber
\end{eqnarray}
\end{theorem}

\begin{remark}
The convergence rate for the unconstrained model is $p_1^{\delta_1} p_2^{\delta_2} T^{-1/2}$ in \cite{wang2018factor}. The rates for the constrained model under different relations between $m_1m_2$ and $p_1p_2$ are shown in Table \ref{table:conv_rate_spdist}. 
\begin{table}[htpb!]
\centering
\resizebox{0.9\textwidth}{!}{
  \begin{tabular}{c|c|c|c} \hline
    & $m_1m_2 \asymp p_1p_2$ & $p_1^{1-\delta_1}p_2^{1-\delta_2} = o(m_1m_2)$ & $m_1m_2 = o(p_1^{1-\delta_1}p_2^{1-\delta_2})$ \\ \hline
   rate & $p_1^{\delta_1} p_2^{\delta_2} T^{-1/2}$ & $m_1 m_2 p_1^{-1} p_2^{-1} \cdot p_1^{\delta_1} p_2^{\delta_2} T^{-1/2}$ & $T^{-1/2}$ \\ \hline
  \end{tabular} }
\caption{Convergence rate of the loading space estimator.}
\label{table:conv_rate_spdist} 
\end{table}

For strong factors with $\delta_1 = \delta_2 = 0$, the convergence rates are the same for the constrained and unconstrained models. However, $m_1 p_1^{-1+\delta_1} m_2 p_2^{-1+\delta_2} T^{-1/2} = o(1)$ is automatically satistied since $p_i >> m_i$. Also, the constrained models have much smaller number of parameters, hence potentially have higher efficiency. 

For weak factors, the constrained models have better convergence rate in most cases.  
%The rate of convergence in Theorem \ref{theorem:conv_rate_loading_mat} 
It depends on the growth rate of the ratio between $m_1 m_2$ and $p_1^{1-\delta_1} p_2^{1-\delta_2}$. The smaller the ratio, the faster the convergence rate. It can be viewed as strength gained due to the constraints. 
%When $p_1^{1-\delta_1}p_2^{1-\delta_2} = O(m_1m_2)$, the ratio of the convergence rates between the constrained and unconstrained models is of the order of $m_1 m_2 p_1^{-1} p_2^{-1}$. 
For example, when $m_1= p_1^{\alpha_1}$ and $m_2 = p_2^{\alpha_2}$, the convergence rate is $p_1^{\delta_1+\alpha_1-1} p_2^{\delta_2+\alpha_2-1} T^{-1/2}$, and we achieve a better rate than that of the unconstrained case if $\alpha_1 < 1$ or $\alpha_2 < 1$. It effectively increases the strength from $\delta_1$ and $\delta_2$ to $\delta_1-(1-\alpha_1)$ and $\delta_2-(1-\alpha_2)$, respectively. Hence, the constraints are particularly useful for weak strength cases. 

When $m_1 m_2 = O(p_1^{1-\delta_1} p_2^{1-\delta_2})$, we achieve the optimal rate $O_p \left( T^{-1/2} \right)$. Note the unconstrained model can only achieve this rate in the case of strong factor. The constrained model can achieve the optimal rate even in the weak factor case. 
A special case is when the dimensions of the constrained row and column loading spaces $m_1$ and $m_2$ are fixed, the convergence rate is $T^{-1/2}$ regardless of the factor strength condition.
Increasing $p_1$ or $p_2$ while keeping $m_1$ and $m_2$ fixed amounts to increasing the sample points in the constrained spaces. When the constrained spaces are properly specified, the additional information introduced from more sample points will accrue and translate into the transformed signal part in (\ref{eqn:cmfm_trans_orth}), while the transformed noise gets canceled out by averaging. 
%The noise-signal ratio $\frac{m_1m_1}{p_1^{1-\delta_1}p_2^{1-\delta_2}}$ goes to zero. 
However, the convergence rate is still bounded below by the convergence rate of the estimated covariance matrix.
When $m_1m_2 \asymp p_1p_2$, the convergence rates of the constrained and unconstrained models are the same. A special case is when $m_1= c_1 p_1$ and $m_2 = c_2 p_2$, that is, the dimensions of the constrained loading spaces increase with $p$'s linearly.
%The improvement on the convergence rate for the weak factor case results from the fact that the constraints effectively reduce the dimension of the data matrix from $p_1 \times p_2$ to $m_1 \times m_2$ whereas the error magnitude remains the same. 
%In the case of weak factor, when $m_1 m_2$ increases faster than $p_1^{1-\delta_1} p_2^{1-\delta_2}$, we have $\norm{ \widehat{\bQ}_i - \bQ_i }_2 = O_p \left( \frac{m_1}{p_1^{1-\delta_1}} \frac{m_2}{p_2^{1-\delta_2}}  T^{-1/2} \right)$.
\end{remark}

\begin{remark}
Under some conditions the convergence rates in Theorem \ref{theorem:conv_rate_loading_mat} may improve significantly. For example, if $\Sigma_u \equiv Var(vec(\bU_t))$ is diagonal (i.e. $U_{t,ij}$ and $U_{t,lk}$ are uncorrelated for $(i,j) \ne (l,k)$) and if we have the grouping constraints (i.e. $\bH_R^{(1)}$ in Section \ref{sec:constraints}), then each elements in $\bE_t$ is a group average. $Var(\bE_{t,ij})$ is smaller by a factor of $\frac{m_1m_2}{p_1p_2}$ and goes to zero when $\frac{m_1m_2}{p_1p_2} = o(1)$.
\end{remark}

\begin{remark}
%The condition $m_1 p_1^{-1+\delta_1} m_2 p_2^{-1+\delta_2} T^{-1/2} = o(1)$ is necessary for the sample autocovariance matrix $\widehat{\mathbf{\Omega}}_{x, ij}(h)$ in (\ref{eqn:M_def}) to be a consistent estimator. Constraints on the loading matrix will lead to some sparsity condition on the matrix M defined in (\ref{eqn:M_def}). 
If the constraints are correct, it would induce certain intrinsic sparsity in the auto-cross-correlation matrix under the unconstrained model. 
For such intrinsic sparsity conditions, we may instead use thresholding estimator for large covariance matrix by \cite{bickel2008covariance} in (\ref{eqn:M_def}). This will lead to faster convergence rates. See Section 3.2 of \cite{chang2018principal}.
By explicitly incorporating constraints in the model, the loading matrix is condensed and the sparsity issue becomes less serious. 
\end{remark}

\begin{remark}
The factors under our definition contain the classic ``common factors'' and the serially correlated idiosyncratic components. As shown by the theoretical properties and simulation studies, the constrained matrix factor model helps identify the weak factors. However, the method is still limited in the sense that it can only identify ``common" factors of some strength $\delta < 1$. In the case of $\delta=1$, although the loading spaces can still be consistently estimated with very large $T$ ($pT^{-1/2}=o(1)$), the factor itself cannot be consistently estimated. Therefore, serially correlated idiosyncratic components for which $\delta=1$ are left ``erroneously" in the noise in application. Hopefully, the constraints may improve the effective factor strength.
\end{remark}

%\begin{remark}
%The strengths of row factors and column factors $\delta_1$ and $\delta_2$ determine the convergence rate jointly. An increase in the strength of row factors is able to improve the estimation of the column factors loading space and vice versa. 
%\end{remark}

\begin{theorem} \label{theorem:conv_rate_eigval_M}
Under Conditions 1-6, and if $m_1 p_1^{-1+\delta_1} m_2 p_2^{-1+\delta_2} T^{-1/2} = o(1)$ and the 
$\bM$ matrix in (\ref{eqn:M_def}) has $k_1$ distinct positive eigenvalues, then the eigenvalues $\{ \hat{\lambda}_1, \ldots, \hat{\lambda}_{m_1} \}$ of $\widehat{\bM}$, sorted in the descending order, satisfy 
\begin{eqnarray}
| \hat{\lambda}_j - \lambda_j | & = & O_p \left( \max \left( p_1^{2-2\delta_1} p_2^{2-2\delta_2}, \; m_1 p_1^{1-\delta_1} m_2 p_2^{1-\delta_2} \right) \cdot T^{-1/2} \right), \quad for \quad j = 1, 2, \ldots, k_1,  \nonumber \\
\qquad | \hat{\lambda}_j | & = & O_p \left( \max \left( p_1^{2-2\delta_1} p_2^{2-2\delta_2}, \; m_1^2 m_2^2 \right) \cdot T^{-1} \right), \quad for \quad j = k_1+1, \ldots, m_1,  \nonumber
\end{eqnarray}
where $\lambda_1 > \lambda_2 > \cdots > \lambda_{m_1}$ are the eigenvalues of $\bM$.
\end{theorem}

Theorem \ref{theorem:conv_rate_eigval_M} shows that the estimators of the nonzero eigenvalues of $\bM$ converge more slowly than those of the zero eigenvalues. 
This provides the theoretical support for the ratio-based estimator of the number of factors described in Section \ref{sec:estimation:subsec:orth_const}. 
The assumption that $\bM$ has $k_1$ distinct positive eigenvalues is not essential, yet it substantially simplifies the presentation and the proof of the convergence properties. 

The convergence rates for the unconstrained model are $\Delta^{\lambda}_{pT} \equiv p_1^{2-\delta_1}p_2^{2-\delta_2}T^{-1/2}$ for the non-zero eigenvalues and $p_1^{\delta_1}p_2^{\delta_2} T^{-1/2} \cdot \Delta^{\lambda}_{pT}$ for the zero eigenvalues, respectively. See \cite{wang2018factor}. The rates for the constrained model under different relations between $m_1m_2$ and $p_1p_2$ are shown in Table \ref{table:conv_rate_eigval}. 

\begin{table}[htbp!]
\centering
\resizebox{0.9\textwidth}{!}{
  \begin{tabular}{c|c|c|c} \hline
    $\lambda_j$ & $m_1m_2 \asymp p_1p_2$ & $p_1^{1-\delta_1}p_2^{1-\delta_2} = o(m_1m_2)$ & $m_1m_2 = o(p_1^{1-\delta_1}p_2^{1-\delta_2})$ \\ \hline    
    Zero & $p_1^{\delta_1}p_2^{\delta_2} T^{-1/2} \cdot \Delta^{\lambda}_{pT}$ & $(\frac{m_1 m_2}{p_1 p_2})^2 p_1^{\delta_1}p_2^{\delta_2} T^{-1/2} \cdot \Delta^{\lambda}_{pT}$ & $p_1^{-\delta_1} p_2^{-\delta_2} T^{-1/2} \cdot \Delta^{\lambda}_{pT}$ \\ \hline
    Non-zero & $\Delta^{\lambda}_{pT}$ & $m_1 m_2 p_1^{-1} p_2^{-1} \cdot \Delta^{\lambda}_{pT}$ & $p_1^{-\delta_1} p_2^{-\delta_2} \cdot \Delta^{\lambda}_{pT}$ \\ \hline \hline
    Ratio & $p_1^{\delta_1}p_2^{\delta_2} T^{-1/2}$ & $m_1 m_2 p_1^{-1+\delta_1} p_2^{-1+\delta_2} T^{-1/2}$ & $T^{-1/2}$ \\ \hline
  \end{tabular} }
\caption{Convergence rate of estimators for non-zero and zero eigenvalues of $\bM$.}
\label{table:conv_rate_eigval}
\end{table}

In the cases of strong factors or weak factors with $m_1 m_2 \asymp p_1p_2$, our result is the same as that of \cite{wang2018factor}. In all other cases, the gap between the convergence rates of nonzero and zero eigenvalues of $\bM$ is larger in the constrained case. 
%When the factors are weak and $p_1^{1-\delta_1} p_2^{1-\delta_2} \sim O_p(m_1 m_2)$, the gap between the convergence rates of nonzero and zero eigenvalues of $\bM$ is larger in the constrained case. 

Let $\bS_t$ be the dynamic signal part of $\bY_t$, i.e. $\bS_t = \bH_R \bR \bF_t \bC' \bH_C' = \bH_R \bQ_1 \bZ_t \bQ_2' \bH_C'$. From the discussion in Section \ref{sec:estimation:subsec:orth_const}, $\bS_t$ can be estimated by 
\begin{equation}
\widehat{\bS}_t = \bH_R \widehat{\bQ}_1 \widehat{\bZ}_t \widehat{\bQ}_2' \bH_C'. \nonumber
\end{equation}
Some theoretical properties of $\widehat{\bS}_t$ are given below: 

\begin{theorem}  \label{theorem:conv_rate_signal}
Under Conditions 1-6 and $m_1 p_1^{-1+\delta_1} m_2 p_2^{-1+\delta_2} T^{-1/2} = o(1)$, we have
\begin{eqnarray}	
\frac{1}{\sqrt{p_1 p_2}} \norm{\widehat{\bS}_t - \bS_t}_2 & = & O_p \left( \max \left( p_1^{-\delta_1/2} p_2^{-\delta_2/2}, \quad m_1 p_1^{-1+\delta_1/2} m_2 p_2^{-1+\delta_2/2} \right)\cdot \frac{1}{\sqrt{T}} + \frac{1}{\sqrt{p_1 p_2}} \right), \nonumber \\ 
& = & \left\{ \begin{array}{l}
O_p \left( p_1^{-\delta_1/2} p_2^{-\delta_2/2} T^{-1/2} + p_1^{-1/2} p_2^{-1/2} \right), \; if \; m_1m_2 = O_p(p_1^{1-\delta_1} p_2^{1-\delta_2}), \\
O_p \left( m_1 p_1^{-1+\delta_1/2} m_2 p_2^{-1+\delta_2/2} T^{-1/2} + p_1^{-1/2} p_2^{-1/2} \right), \; otherwise.
\end{array} \right. \nonumber
\end{eqnarray}
\end{theorem}

%When $m_1m_2 \sim O_p(p_1^{1-\delta_1} p_2^{1-\delta_2})$, the rate in Theorem \ref{theorem:conv_rate_signal} becomes  and increases in $p_1$ and $p_2$ improve the convergence rate. In other cases, we get $\frac{1}{\sqrt{p_1 p_2}} \norm{\widehat{\bS}_t - \bS_t}_2 = O_p \left( m_1 p_1^{-1+\delta_1/2} m_2 p_2^{-1+\delta_2/2} T^{-1/2} + p_1^{-1/2} p_2^{-1/2} \right)$. 
Theorem 3 shows that as long as $m_1m_2 = o(p_1p_2)$ we achieve for weak factor cases a faster convergence rate than $O_p \left(p_1^{\delta_1/2} p_2^{\delta_2/2} T^{-1/2} + p_1^{-1/2} p_2^{-1/2} \right)$ -- the convergence rate of the unconstrained model in \cite{wang2018factor}. When $m_1m_2 = O_p(p_1^{1-\delta_1} p_2^{1-\delta_2})$, we get an even better rate. Note that the estimation of the loading spaces are consistent with fixed $p_1$ and $p_2$ in Theorem 1. But the consistency of the signal estimate requires $p_1, p_2 \rightarrow \infty$.

Asymptotic theories for estimators of nonorthogonal, multi-term constrained factor models are trivial extensions of the above properties for the orthogonal constrained factor model. 

\section{Simulation}
\label{sec:simulation}

In this section, we present some simulation study to illustrate the performance of the estimation methods of Section \ref{sec:estimation} in finite samples. We also compare the results with those of unconstrained models. We employ data generating models with orthogonal full and partial constraints, respectively. In the simulation, we use the Student-$t$ distribution with $5$ degrees of freedom to generate the entries in the disturbances $\bU_t$. Using Gaussian noises shows similar results. 

As noted in Section \ref{sec:estimation}, the row and column factor loading matrices $\bLambda = \bH_R \bR$ and $\bGamma = \bH_C \bC$ are only identifiable up to a linear space spanned by its columns. Following \cite{Lam-Yao-2011} and \cite{wang2018factor}, we adopt the discrepancy measure used by \cite{chang2015high}: for two orthogonal matrices $\bO_1$ and $\bO_2$ of size $p \times q_1$ and $p \times q_2$, then the difference between the two linear spaces $\mathcal{M}(\bO_1)$ and $\mathcal{M}(\bO_2)$ is measured by 
\begin{equation}  \label{eqn:space_dist_def}
\mathcal{D}(\mathcal{M}(\bO_1), \mathcal{M}(\bO_2)) = \left( 1 - \frac{1}{\max (q_1, q_2)} tr \left( \bO_1 \bO_1' \bO_2 \bO_2' \right) \right) ^{1/2}. 
\end{equation}

Clearly, $\mathcal{D}(\mathcal{M}(\bO_1), \mathcal{M}(\bO_2))$ assumes values in [0,1]. It equals to $0$ if and only if $\mathcal{M}(\bO_1) = \mathcal{M}(\bO_2)$ and equals to $1$ if and only if $\mathcal{M}(\bO_1) \perp \mathcal{M}(\bO_2)$. If $\bO_1$ and $\bO_2$ are vectors, (\ref{eqn:space_dist_def}) is the cosine similarity measure. We report this space distance $\mathcal{D}(\cdot, \cdot)$ as a measurement of the discrepency between estimated and true loading spaces. 

\subsection{Case 1. Orthogonal Constraints}   \label{sec:simulation:subsec:orth_cons}

In this case, the observed data $\bY_t$'s are generated according to Model (\ref{eqn:cmfm}), 
\begin{equation*}
\bY_t = \bH_R \bR \bF_t \bC' \bH_C' + \bU_t, \qquad t = 1, \ldots, T,
\end{equation*} 
under the following design.

The latent factor process $\bF_t$ is of dimension $k_1 \times k_2 = 3 \times 2$. The entries of $\bF_t$ follow $k_1 k_2$ independent $AR(1)$ processes with Gaussian white noise $\mathcal{N}(0,1)$ innovations. Specifically, $vec(\bF_t) = \bPhi_F \, vec(\bF_{t-1}) + \boldsymbol{\epsilon}_t$ with $\bPhi_F = diag(-0.5, 0.6, 0.8, -0.4, 0.7, 0.3)$. The dimensions of the constrained row and column loading spaces are $m_1=12$ and $m_2=3$, respectively. Hence, $\bR$ is  $12 \times 3$ and $\bC$ is $3 \times 2$. The entries of $\bR$ and $\bC$ are independently sampled from the uniform distribution $U(- p_i^{-\delta_i/2} \sqrt{m_i/p_i}, \; p_i^{-\delta_i/2} \sqrt{m_i/p_i})$ for $i = 1, 2$, respectively, so that the condition on the factor strength is satisfied. The disturbance $\bU_t = \bPsi^{1/2} \bXi_t$ is a white noise process, where the elements of $\bXi_t$ are independent random variables of Student-$t$ distribution with five degrees of freedom and the matrix $\bPsi^{1/2}$ is chosen so that $\bU_t$ has a Kronecker product covariance structure $cov(vec(\bU_t)) = \bGamma_2 \otimes \bGamma_1$, where $\bGamma_1$ and $\bGamma_2$ are of size $p_1 \times p_1$ and $p_2 \times p_2$ respectively. For $\bGamma_1$ and $\bGamma_2$, 
the diagonal elements are 1 and the off-diagonal elements are 0.2. 

The effects of factor strength are investigated by varying factor strength parameter $(\delta_1, \delta_2)$ among $(0, 0)$, $(0.5, 0)$, $(0.5, 0.5)$. For each pair of $\delta_i$'s, the dimensions $(p_1, p_2)$ are chosen to be $(20, 20)$, $(20, 40)$, $(40, 20)$ and $(40, 40)$. The sample sizes $T$ are $0.5 p_1 p_2$, $p_1 p_2$, $1.5 p_1 p_2$ and $2 p_1 p_2$. For each combination of the  parameters, we use 500 realizations. And we use $h_0=1$ for all simulations. The estimation error of $\mathcal{M}(\widehat{\bQ}_i)$ is defined as $\mathcal{D}(\widehat{\bQ}_i, \bQ_i)$, where the distance $\mathcal{D}$ is defined in (\ref{eqn:space_dist_def}). 

The row constraint matrix $\bH_R$ is a $p_1 \times 12$ orthogonal matrix. For $p_1=20$, $\bH_R$ is assumed to be a block diagonal matrix $\bI_4 \otimes \bD$, where $\bI_k$ is the identify matrix of dimension $k$ and $\bD = [\utwi{d}_1, \utwi{d}_2, \utwi{d}_3]$ is a $5 \times 3$ matrix with $\utwi{d}_1'=(1,1,1,1,1)/\sqrt{5}$, $\utwi{d}_2'=(-1,-1,0,1,1)/2$, $\utwi{d}_3'=(-1,0,2,0,-1)/\sqrt{6}$. These three $\utwi{d}_j$ vectors can be viewed as the level, slope and curvature, respectively, of a group of five variables. Therefore, the 20 rows are divided into 4 groups of size 5. When we increase $p_1$ to $40$ while keeping $m_1 = 12$ fixed, we double the length of each vector in the columns of $\bD$, using $\utwi{d}_1'=(1,1,1,1,1,1,1,1,1,1)/\sqrt{10}$, $\utwi{d}_2'=(-1,-1,-1,-1,0,0,1,1,1,1)/\sqrt{8}$ and  $\utwi{d}_3'=(-1,-1,0,0,2,2,0,0,-1,-1)/\sqrt{12}$. %Therefore, in this situation, the number of row variables in each of the four groups doubles and the estimation of the group level, the slope and the curvature of the row variables is refined by using all $10$ row variables in each group. 

The column constraint matrix $\bH_C$ is a $p_2 \times 3$ orthogonal matrix. For $p_2=20$, the three columns of $\bH_C$ are generated as $\bh_{c,1} = [\mathbf{1}_7 / \sqrt{7}, \mathbf{0}_7, \mathbf{0}_6]'$, $\bh_{c,2} = [\mathbf{0}_7, \mathbf{1}_7 / \sqrt{7}, \mathbf{0}_6]'$, $\bh_{c,3} = [\mathbf{0}_7, \mathbf{0}_7, \mathbf{1}_6 / \sqrt{6}]'$, where $\mathbf{0}_k$ denotes a $k$-dimensional zero row vector. The constraints represent a 3-group  classification. The 20 columns are divided into 3 groups of size 7, 7 and 6, respectively. In increasing $p_2$ to $40$ while keeping $m_2 = 3$ fixed, we double the length of each vector in the columns defined above. %The number of variables we use to estimate the loading on group factor doubles in every group. 

Table \ref{table:orth_k_right} shows the performance of estimating the true number of row and column factors seperately. The subscripts \textit{c} and \textit{u} denote results from the constrained model (\ref{eqn:cmfm}) and unconstrained model (\ref{eqn:mfm}), 
respectively. $f_{c}$ and $f_{u}$ denote the relative frequency of correctly estimating the true number of factors $k$. From the table, we make the following observations. First, when the 
row and column factors are strong, i.e. $(\delta_1,\delta_2) = (0,0)$, both constrained and 
unconstrained models can estimate accurately the number of factors, but the constrained models 
perform better when the sample size is small. Second, if the strength of the row factors is weak, but 
the strength of the column factors is strong, i.e. $(\delta_1,\delta_2) = (0.5,0)$, the 
unconstrained models fail to estimate the number of row factors $k_1$, but the constrained models 
continue to perform well for both $k_1$ and $k_2$. Furthermore, as expected, the performance of the constrained models 
improves with the sample size. Finally, if the strength of the row and columns factors is weak, 
i.e. $(\delta_1,\delta_2) = (0.5,0.5)$, both models encounter difficulties in estimating the 
correct number of row factors $k_1$ for the sample sizes used. However, the constrained models continue to perform well for the number of column factors $k_2$. Here $m_1$ and $m_2$ are different and play a role. Since $m_1 > m_2$, $k_2$ is estimated with higher accuracy, especially in the weak factor case.
\begin{table}[htpb!]
	\centering
	\resizebox{0.8\textwidth}{!}{%
\begin{tabular}{cccc|cc|cc|cc|cc}
\hline
\multicolumn{4}{c|}{$k_1$} & \multicolumn{2}{c|}{$T = 0.5 \, p_1 \, p_2$} & \multicolumn{2}{c|}{$T =  \, p_1 \, p_2$} & \multicolumn{2}{c|}{$T = 1.5 \, p_1 \, p_2$} & \multicolumn{2}{c}{$T = 2 \, p_1 \, p_2$} \\ \hline
$\delta_1$ & $\delta_2$ & $p_1$ & $p_2$ & $f_{u}$ & $f_{c}$ & $f_{u}$ & $f_{c}$ & $f_{u}$ & $f_{c}$ & $f_{u}$ & $f_{c}$ \\ \hline
\multirow{4}{*}{0} & \multirow{4}{*}{0} & 20 & 20 & 0.264 & 0.942 & 0.728 & 0.996 & 0.952 & 1 & 0.996 & 1 \\
 &  & 20 & 40 & 0.734 & 1 & 0.998 & 1 & 1 & 1 & 1 & 1 \\
 &  & 40 & 20 & 0.786 & 0.994 & 1 & 1 & 1 & 1 & 1 & 1 \\
 &  & 40 & 40 & 1 & 1 & 1 & 1 & 1 & 1 & 1 & 1 \\ \hline
\multirow{4}{*}{0.5} & \multirow{4}{*}{0} & 20 & 20 & 0 & 0.202 & 0 & 0.508 & 0 & 0.734 & 0 & 0.926 \\
 &  & 20 & 40 & 0 & 0.692 & 0 & 0.97 & 0 & 0.994 & 0 & 1 \\
 &  & 40 & 20 & 0 & 0.352 & 0 & 0.77 & 0 & 0.908 & 0 & 0.964 \\
 &  & 40 & 40 & 0 & 0.872 & 0 & 0.984 & 0 & 0.996 & 0 & 0.994 \\ \hline
\multirow{4}{*}{0.5} & \multirow{4}{*}{0.5} & 20 & 20 & 0 & 0.052 & 0 & 0.036 & 0 & 0.01 & 0 & 0.006 \\
 &  & 20 & 40 & 0 & 0.052 & 0 & 0.022 & 0 & 0.008 & 0 & 0.004 \\
 &  & 40 & 20 & 0 & 0.062 & 0 & 0.006 & 0 & 0 & 0 & 0.002 \\
 &  & 40 & 40 & 0 & 0.018 & 0 & 0.006 & 0 & 0.006 & 0 & 0.044 \\ \hline 
\hline
\multicolumn{4}{c|}{$k_2$} & \multicolumn{2}{c|}{$T = 0.5 \, p_1 \, p_2$} & \multicolumn{2}{c|}{$T =  \, p_1 \, p_2$} & \multicolumn{2}{c|}{$T = 1.5 \, p_1 \, p_2$} & \multicolumn{2}{c}{$T = 2 \, p_1 \, p_2$} \\ \hline
$\delta_1$ & $\delta_2$ & $p_1$ & $p_2$ & $f_{u}$ & $f_{c}$ & $f_{u}$ & $f_{c}$ & $f_{u}$ & $f_{c}$ & $f_{u}$ & $f_{c}$ \\ \hline
\multirow{4}{*}{0} & \multirow{4}{*}{0} & 20 & 20 & 1 & 1 & 1 & 1 & 1 & 1 & 1 & 1 \\
 &  & 20 & 40 & 1 & 1 & 1 & 1 & 1 & 1 & 1 & 1 \\
 &  & 40 & 20 & 1 & 1 & 1 & 1 & 1 & 1 & 1 & 1 \\
 &  & 40 & 40 & 1 & 1 & 1 & 1 & 1 & 1 & 1 & 1 \\ \hline
\multirow{4}{*}{0.5} & \multirow{4}{*}{0} & 20 & 20 & 0.988 & 0.996 & 1 & 1 & 1 & 1 & 1 & 1 \\
 &  & 20 & 40 & 1 & 1 & 1 & 1 & 1 & 1 & 1 & 1 \\
 &  & 40 & 20 & 1 & 1 & 1 & 1 & 1 & 1 & 1 & 1 \\
 &  & 40 & 40 & 1 & 1 & 1 & 1 & 1 & 1 & 1 & 1 \\ \hline
\multirow{4}{*}{0.5} & \multirow{4}{*}{0.5} & 20 & 20 & 0.02 & 0.762 & 0.104 & 0.962 & 0.272 & 0.992 & 0.58 & 1 \\
 &  & 20 & 40 & 0 & 0.956 & 0.09 & 0.998 & 0.472 & 1 & 0.85 & 1 \\
 &  & 40 & 20 & 0 & 0.952 & 0.026 & 1 & 0.196 & 1 & 0.572 & 1 \\
 &  & 40 & 40 & 0 & 1 & 0.03 & 1 & 0.438 & 1 & 0.906 & 1 \\ \hline
\end{tabular} %
}
	\caption{Relative frequencies of correctly estimating the number of row (column) factors $k_1$ ($k_2$) in the case of orthogonal constraints, where $p_i$ are the dimension, $T$ is the sample size, 
	and $f_u$ and $f_c$ denote the results of unconstrained and 
	constrained factor model, respectively. Table on top is for $k_1$, while table on the bottom is for $k_2$.}
	\label{table:orth_k_right}
\end{table}

Figure~\ref{fig:orth_Q_mean_sd} shows the box-plots of the estimation errors in estimating the loading spaces of $\bQ = \bQ_2 \otimes \bQ_1$ using the correct number of factors. 
The gray boxes are for the constrained models. 
From the plots, it is seen that when both row and column factors are strong, i.e. $(\delta_1, \delta_2)=(0, 0)$, and the number of factors is properly estimated, 
the mean and standard deviation of the estimation errors $\mathcal{D}(\widehat{\bQ}, \bQ)$ are small for both models, but the constrained model has a smaller mean estimation error. When row factors are weak, i.e. $(\delta_1, \delta_2)=(0.5, 0)$, and the true number of factors is used, 
the estimation error of constrained models remains small whereas that of the unconstrained models 
is substantially larger.  %When both column and row factors are weak, i.e. $(\delta_1, \delta_2)=(0.5, 0.5)$, both models fail at estimating the number of factors with the sample size used. However, with known number of factors, the constrained factor model is still able to estimate the loading space with satisfactory accuracy when the dimension $p$ and observation $T$ are large, while the unconstrained model fails at all situations. % Table \ref{table:orth_Q_mean_sd} in Appendix \ref{appendix:tables_simulation_results} reports mean and standard deviation of the errors. It also shows that the constrained factor models outperform their unconstrained counterparts significantly. 

Table \ref{table:orth_Q1Q2con_mean_sd} shows the mean and standard deviations of the estimation errors $\mathcal{D}(\widehat{\bQ}_i, \bQ_i)$ for row ($i=1$) and column ($i=2$) loading spaces separately for the constrained model (\ref{eqn:cmfm}). 
Column loading spaces are estimated with higher accuracy because the dimension of the constrained column space ($m_2$) is smaller than that of the constrained row space ($m_1$).
Intuitively, after transformation (\ref{eqn:cmfm_trans_orth}), the ratio of the effective column factor strength $\norm{\bC}^2_2$ and noise level $\norm{\bE_t}^2_2$ is larger than the ratio of the effective row factor strength $\norm{\bR}^2_2$ and noise level $\norm{\bE_t}^2_2$. 
%This is consistent with theoretical results in Section \ref{sec:theory}. 
From the table, we see that (a) the mean of estimation errors decreases, as expected, as the sample size increases and (b) the mean of estimation errors is inversely proportional to the strength of row factors. 

\begin{figure}[ht!]
	\centering
	\includegraphics[width=\linewidth,height=\textheight,keepaspectratio=true]{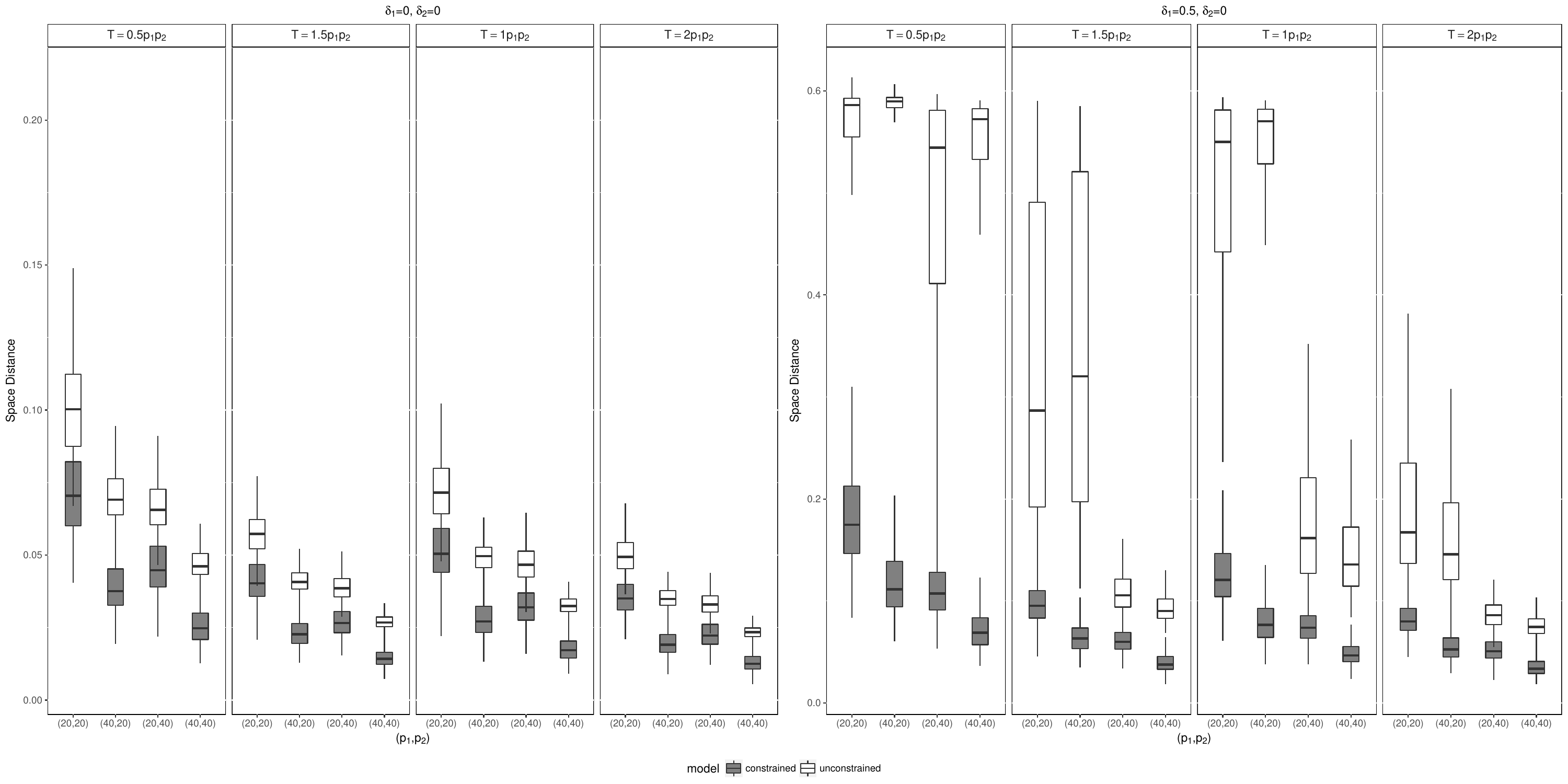}
	\caption{Box-plots of the estimation accuracy measured by $\mathcal{D}(\widehat{\bQ}, \bQ)$ of equation (\ref{eqn:space_dist_def}) for the case of orthogonal constraints. Gray boxes represent the constrained model. The results are based on $500$ iterations. See Table 17 %\ref{table:orth_Q_mean_sd} 
	in Appendix C %\ref{appendix:tables_simulation_results} 
	for plotted values. }
	\label{fig:orth_Q_mean_sd}
\end{figure} 

% Table from "simu.msd.m_12x3.k_3x2_edited.csv"
\begin{table}[ht!]
	\centering
	\resizebox{\textwidth}{!}{%
		\begin{tabular}{cccc|cc|cc|cc|cc}
			\hline
			\multicolumn{4}{c|}{} & \multicolumn{2}{c|}{$ T = 0.5 \, p_1 \, p_2 $} & \multicolumn{2}{c|}{$ T = p_1 \, p_2 $} & \multicolumn{2}{c|}{$ T = 1.5 \, p_1 \, p_2 $} & \multicolumn{2}{c}{$ T = 2 \, p_1 \, p_2 $} \\ \hline
			$\delta_1$ & $\delta_2$ & $p_1$ & $p_2$ & $\mathcal{D}( \widehat{Q}_1, Q_1 )$ & $\mathcal{D}( \widehat{Q}_2, Q_2 )$ & $\mathcal{D}( \widehat{Q}_1, Q_1 )$ & $\mathcal{D}( \widehat{Q}_2, Q_2 )$ & $\mathcal{D}( \widehat{Q}_1, Q_1 )$ & $\mathcal{D}( \widehat{Q}_2, Q_2 )$ & $\mathcal{D}( \widehat{Q}_1, Q_1 )$ & $\mathcal{D}( \widehat{Q}_2, Q_2 )$ \\ \hline
			\multirow{4}{*}{0} & \multirow{4}{*}{0} & 20 & 20 & 0.71(0.18) & 0.13(0.07) & 0.51(0.13) & 0.09(0.05) & 0.41(0.09) & 0.07(0.04) & 0.35(0.07) & 0.06(0.03) \\
			&  & 20 & 40 & 0.46(0.11) & 0.08(0.04) & 0.32(0.07) & 0.05(0.03) & 0.27(0.06) & 0.04(0.02) & 0.23(0.05) & 0.04(0.02) \\
			&  & 40 & 20 & 0.40(0.12) & 0.07(0.04) & 0.28(0.07) & 0.05(0.03) & 0.23(0.06) & 0.04(0.02) & 0.19(0.05) & 0.04(0.02) \\
			&  & 40 & 40 & 0.26(0.07) & 0.04(0.02) & 0.18(0.04) & 0.03(0.02) & 0.14(0.04) & 0.03(0.01) & 0.13(0.03) & 0.02(0.01) \\ \hline
			\multirow{4}{*}{0.5} & \multirow{4}{*}{0} & 20 & 20 & 1.84(0.75) & 0.5(0.23) & 1.23(0.35) & 0.30(0.15) & 0.95(0.23) & 0.22(0.11) & 0.81(0.18) & 0.17(0.09) \\
			&  & 20 & 40 & 1.08(0.30) & 0.26(0.13) & 0.74(0.18) & 0.15(0.08) & 0.61(0.14) & 0.12(0.06) & 0.52(0.12) & 0.10(0.05) \\
			&  & 40 & 20 & 1.18(0.45) & 0.28(0.15) & 0.78(0.23) & 0.17(0.09) & 0.64(0.18) & 0.13(0.07) & 0.54(0.14) & 0.11(0.06) \\
			&  & 40 & 40 & 0.71(0.21) & 0.14(0.08) & 0.48(0.13) & 0.09(0.05) & 0.39(0.1) & 0.07(0.04) & 0.35(0.09) & 0.06(0.03) \\ \hline
			\multirow{4}{*}{0.5} & \multirow{4}{*}{0.5} & 20 & 20 & 5.84(0.62) & 2.04(0.53) & 5.35(0.75) & 1.63(0.42) & 4.68(1.17) & 1.33(0.34) & 4.20(1.31) & 1.13(0.32) \\
			&  & 20 & 40 & 5.62(0.68) & 1.98(0.40) & 4.75(1.13) & 1.47(0.30) & 3.96(1.33) & 1.18(0.27) & 3.32(1.35) & 0.97(0.24) \\
			&  & 40 & 20 & 5.53(0.61) & 1.52(0.50) & 4.68(1.25) & 1.00(0.37) & 3.64(1.46) & 0.76(0.30) & 2.87(1.42) & 0.61(0.25) \\
			&  & 40 & 40 & 5.01(1.01) & 1.32(0.38) & 3.64(1.47) & 0.84(0.29) & 2.62(1.46) & 0.61(0.20) & 1.98(1.14) & 0.49(0.19) \\ \hline
		\end{tabular}%
	}
	\caption{Means and standard deviations (in parentheses) of the estimation accuracy measured by $\mathcal{D}(\widehat{\bQ}, \bQ)$ of equation (\ref{eqn:space_dist_def}) for constrained factor 
	models. The case of orthogonal constraints is used. The subscripts 1 and 2 denote row and 
	column, respectively.
	All numbers in the table are 10 times of the true numbers for clear presentation. The results are based on $500$ simulations. }
	\label{table:orth_Q1Q2con_mean_sd}
\end{table}

To investigate the performance of estimation under different choices of $h_0$, which is the number of lags used in (\ref{eqn:M_def}), we change the underlying generating model of $vec(\bF_t)$ to a VAR(2) process without the lag-1 term, $vec(\bF_t) = \bPhi_F vec(\bF_{t-2}) + \boldsymbol{\epsilon}_t$. Here we only consider the strong factor setting with $\delta_1 = \delta_2 = 0$ and use the sample size $T = 2 p_1 p_2$ for each combination of $p_1$ and  $p_2$. All the other parameters are the same as those in % Section \ref{sec:simulation:subsec:orth_cons}. 
the prior simulation. 
Table \ref{table:AR(2)_h0_change} presents the simulation results. Since $vec(\bF_t)$, and hence $vec(\bY_t)$, has zero auto-covariance matrix at lag $1$, $\widehat{\bM}$ under $h_0=1$ contains no information on the signal, and, as expected, both the constrained and unconstrained models fail to correctly estimate the number of factors and the loading space.  On the other hand, both models 
are able to correctly estimate the number of factors when $h_0 > 1$ with the constrained model 
faring better. 
The fact that $h_0=2, 3, 4$ give similar results shows that the choice of $h_0$ is not critical to the performance of the proposed method as long as it is sufficiently large to describe the pattern of the auto-covariance matrices of the data. See Condition 2 in Appendix A. In practice, one can select $h_0$  by examining the sample 
cross-correlation matrices of $\bY_t$.     
%%%
\begin{table}[ht!]
	\centering
	\resizebox{0.8\textwidth}{!}{%
		\begin{tabular}{c|cc|cccc}
			\hline
			& $p_1$ & $p_2$ & $h_0=1$ & $h_0=2$ & $h_0=3$ & $h_0=4$ \\ \hline \hline
			\multirow{4}{*}{$f_{c}$} & 20 & 20 & 0.12 & 1.00 & 1.00 & 1.00 \\
			& 20 & 40 & 0.16 & 1.00 & 1.00 & 1.00 \\
			& 40 & 20 & 0.12 & 1.00 & 1.00 & 1.00 \\
			& 40 & 40 & 0.22 & 1.00 & 1.00 & 1.00 \\ \hline
			\multirow{4}{*}{$f_{u}$} & 20 & 20 & 0.00 & 0.89 & 0.58 & 0.43 \\
			& 20 & 40 & 0.00 & 1.00 & 1.00 & 0.95 \\
			& 40 & 20 & 0.00 & 1.00 & 1.00 & 0.97 \\
			& 40 & 40 & 0.00 & 1.00 & 1.00 & 1.00 \\ \hline \hline
			\multirow{4}{*}{$\mathcal{D}_{c}(\widehat{\bQ}, \bQ)$} & 20 & 20 & 2.83(1.13) & 0.36(0.07) & 0.37(0.07) & 0.38(0.08) \\
			& 20 & 40 & 2.69(1.15) & 0.23(0.05) & 0.23(0.05) & 0.24(0.05) \\
			& 40 & 20 & 2.54(1.21) & 0.20(0.05) & 0.20(0.05) & 0.21(0.06) \\
			& 40 & 40 & 2.31(1.17) & 0.13(0.03) & 0.13(0.03) & 0.14(0.04) \\ \hline
			\multirow{4}{*}{$\mathcal{D}_{u}(\widehat{\bQ}, \bQ)$} & 20 & 20 & 4.37(1.29) & 0.51(0.07) & 0.53(0.07) & 0.53(0.08) \\
			& 20 & 40 & 4.30(1.30) & 0.34(0.04) & 0.35(0.04) & 0.35(0.04) \\
			& 40 & 20 & 4.36(1.31) & 0.36(0.04) & 0.37(0.04) & 0.37(0.05) \\
			& 40 & 40 & 4.34(1.34) & 0.24(0.02) & 0.24(0.03) & 0.25(0.03) \\ \hline
		\end{tabular}%
	}
	\caption{Performance of estimation under different choices of $h_0$ when $vec(\bF_t) = \bPhi_{F} vec(\bF_{t-2}) + \boldsymbol{\epsilon}_t$. Metrics reported are relative frequencies of correctly estimating $k$, means and standard deviations (in parentheses) of the estimation accuracy measured by $\mathcal{D}(\widehat{\bQ}, \bQ)$. Means and standard deviations are multiplied by $10$ for ease in presentation. $f_u$ and $f_c$ denote unconstrained and constrained models. }
	\label{table:AR(2)_h0_change}
\end{table}

% % %--------------------------------------------------------------------------
\subsection{Case 2. Partial Orthogonal Constraints}
In this case, the observed data $\bY_t$'s are generated using Model (\ref{eqn:pcmfm_explicit}), 
\begin{equation*}
\bY_t = \bH_R \bR_1 \bF_t \bC_1' \bH_C' + \bL_R \bR_2 \bG_t \bC_2' \bL_C' + \bU_t , \qquad t = 1, \ldots, T.
\end{equation*} 
Parameter settings of the first part $\bH_R \bR_1 \bF_t \bC_1' \bH_C'$ are the same as those in Case 1. The latent factor process $\bG_t$ is of dimension $q_1 \times q_2 = 5 \times 4$. The entries of $\bG_t$ follow $q_1q_2$ independent $AR(1)$ processes with Gaussian white noise $\mathcal{N}(0,1)$ innovations, $vec(\bG_t) = \bPhi_G \, vec(\bG_{t-1}) + \boldsymbol{\epsilon}_t$ with $\bPhi_G$ being a diagonal matrix with entries $(-0.7, 0.5, -0.2, 0.9, 0.1, 0.4, 0.6, -0.5, 0.7, 0.7, -0.4, 0.4, 0.4, -0.6, -0.6, \\ 0.6, -0.5, -0.3, 0.2, -0.4)$. The row loading matrix $\bL_R \bR_2$ is a $20 \times 5$ orthogonal matrix, satisfying $\bH_R' \bL_R = \mathbf{0}$. The column loading matrix $\bL_C \bC_2$ is a $20 \times 4$ orthogonal matrix, satisfying $\bH_C' \bL_C = \mathbf{0}$. The entries of $\bR_2$ and $\bC_2$ are random draws from the uniform distribution between $- p_i^{-\eta_i/2} \sqrt{p_i/(p_i-m_i)}$ and $p_i^{-\eta_i/2} \sqrt{p_i/(p_i-m_i)}$ for $i = 1, 2$,  respectively, so that the conditions on factor strength are satisfied. 
Factor strength is controlled by the $\delta_i$'s. 

Model (\ref{eqn:pcmfm_explicit}) could be written in the following form:
\begin{equation*}
\bY_t = (\bH_R \bR_1 \; \bL_R \bR_2) \begin{pmatrix}
                         \bF_t & 0 \\
                         0 & \bG_t
                        \end{pmatrix} \begin{pmatrix} \bC'_1 \bH'_C \\ \bC'_2 \bL'_C \end{pmatrix} + \bU_t, \qquad t = 1, \ldots, T.
\end{equation*} 
In this form, the true number of factors is $k_0 = (k_1 + r_1)(k_2+r_2)$ and the true loading matrix is $ (\bH_C \bC_1 \; \bL_C \bC_2) \otimes (\bH_R \bR_1 \; \bL_R \bR_2) $. Table \ref{table:partial_k1_right_select} shows the frequency of correctly estimating $k_0$ based on $500$ iterations. In the table, $f_{u}$ denotes the frequency of correctly estimating $k_0$ for unconstrained model. $f_{con_1}$ and $f_{con_2}$ denote the same frequency metric for the first matrix factor $\bF_t$ and second matrix factor $\bG_t$ of the constrained model. The number of factors in $\bF_t$ is estimated with a higher accuracy because the dimension of constrained loading space for $\bF_t$ is $m_1 m_2 = 36$, which is smaller than that for $\bG_t$, $(p_1-m_1)(p_2-m_2)=136$. The result again confirms the theoretical results in Section \ref{sec:theory}.  
Note that Table \ref{table:partial_k1_right_select} only contains selected combinations of factor strength parameters $\delta_i$'s ($i=1, \ldots 4$).  The results of all combinations of 
factor strength are given in Table~18 %\ref{table:partial_k1_right} 
in Appendix~C. %\ref{appendix:tables_simulation_results}.

%Table \ref{table:partial_Q_mean_sd_select} compares the mean and standard deviation of the estimation errors for the partially constrained and unconstrained matrix factor models.
Figure \ref{fig:partial_Q_mean_sd_strong} and Figure \ref{fig:partial_Q_mean_sd_weak} present box-plots of estimation errors under weak and strong factors based on $500$ simulations, respectively. Again, the results show that the constrained approach effectively improves the estimation accuracy. The performance of constrained model is good even in the case of weak factors. Moreover, with stronger signals and larger sample sizes, both approaches increase their estimation accuracy. 

% Table from "simu.k1right.m_12x3.k_3x2.q_5x4_edited.csv"
\begin{table}[ht!]
	\centering
	\resizebox{0.85\textwidth}{!}{%
		\begin{tabular}{cccc|cc|ccc|ccc|ccc|ccc}
			\hline
			\multicolumn{6}{c|}{} & \multicolumn{3}{c|}{$T = 0.5 * p_1 * p_2$} & \multicolumn{3}{c|}{$T = p_1 * p_2$} & \multicolumn{3}{c|}{$T = 1.5 * p_1 * p_2$} & \multicolumn{3}{c}{$T = 2 * p_1 * p_2$} \\ \hline
			\multicolumn{1}{c}{$\delta_1$} & \multicolumn{1}{c}{$\delta_2$} & \multicolumn{1}{c}{$\delta_3$} & \multicolumn{1}{c|}{$\delta_4$} & \multicolumn{1}{c}{$p_1$} & \multicolumn{1}{c|}{$p_2$} & \multicolumn{1}{c}{$f_{u}$} & \multicolumn{1}{c}{$f_{con_1}$} & \multicolumn{1}{c|}{$f_{con_2}$} & \multicolumn{1}{c}{$f_{u}$} & \multicolumn{1}{c}{$f_{con_1}$} & \multicolumn{1}{c|}{$f_{con_2}$} & \multicolumn{1}{c}{$f_{u}$} & \multicolumn{1}{c}{$f_{con_1}$} & \multicolumn{1}{c|}{$f_{con_2}$} & \multicolumn{1}{c}{$f_{u}$} & \multicolumn{1}{c}{$f_{con_1}$} & \multicolumn{1}{c}{$f_{con_2}$} \\ \hline
			\multirow{4}{*}{0} & \multirow{4}{*}{0} & \multirow{4}{*}{0} & \multirow{4}{*}{0} & 20 & 20 & 0 & 0.94 & 0 & 0 & 1.00 & 0 & 0 & 1.00 & 0 & 0.01 & 1.00 & 0 \\
			&  &  &  & 20 & 40 & 0 & 1.00 & 0 & 0 & 1.00 & 0 & 0.03 & 1.00 & 0 & 0.19 & 1.00 & 0 \\
			&  &  &  & 40 & 20 & 0.15 & 0.99 & 1.00 & 0.81 & 1.00 & 1.00 & 0.98 & 1.00 & 1.00 & 1.00 & 1.00 & 1.00 \\
			&  &  &  & 40 & 40 & 0.71 & 1.00 & 1.00 & 1.00 & 1.00 & 1.00 & 1.00 & 1.00 & 1.00 & 1.00 & 1.00 & 1.00 \\ \hline
			\multirow{4}{*}{0} & \multirow{4}{*}{0} & \multirow{4}{*}{0.5} & \multirow{4}{*}{0} & 20 & 20 & 0 & 0.94 & 0 & 0 & 1.00 & 0 & 0 & 1.00 & 0 & 0 & 1.00 & 0 \\
			&  &  &  & 20 & 40 & 0 & 1.00 & 0 & 0 & 1.00 & 0 & 0 & 1.00 & 0 & 0 & 1.00 & 0 \\
			&  &  &  & 40 & 20 & 0 & 0.99 & 0.54 & 0 & 1.00 & 0.84 & 0 & 1.00 & 0.97 & 0 & 1.00 & 1.00 \\
			&  &  &  & 40 & 40 & 0 & 1.00 & 0.98 & 0 & 1.00 & 1.00 & 0 & 1.00 & 1.00 & 0 & 1.00 & 1.00 \\ \hline
			\multirow{4}{*}{0.5} & \multirow{4}{*}{0.5} & \multirow{4}{*}{0.5} & \multirow{4}{*}{0.5} & 20 & 20 & 0 & 0.07 & 0 & 0 & 0.04 & 0 & 0 & 0.01 & 0 & 0 & 0.01 & 0 \\
			&  &  &  & 20 & 40 & 0 & 0.07 & 0 & 0 & 0.02 & 0 & 0 & 0.01 & 0 & 0 & 0.01 & 0 \\
			&  &  &  & 40 & 20 & 0 & 0.06 & 0 & 0 & 0.01 & 0 & 0 & 0 & 0 & 0 & 0 & 0 \\
			&  &  &  & 40 & 40 & 0 & 0.06 & 0 & 0 & 0 & 0 & 0 & 0 & 0 & 0 & 0.03 & 0 \\ \hline
		\end{tabular}%
	}
	\caption{Relative frequencies of correctly estimating the number of factors for partially constrained factor models. The full table including all combinations is presented in Table~18 %\ref{table:partial_k1_right} 
	in Appendix~C. %\ref{appendix:tables_simulation_results}.
	}
	\label{table:partial_k1_right_select}
\end{table}

\begin{figure}[ht!]
    \centering
    \includegraphics[width=\linewidth,height=\textheight,keepaspectratio=true]{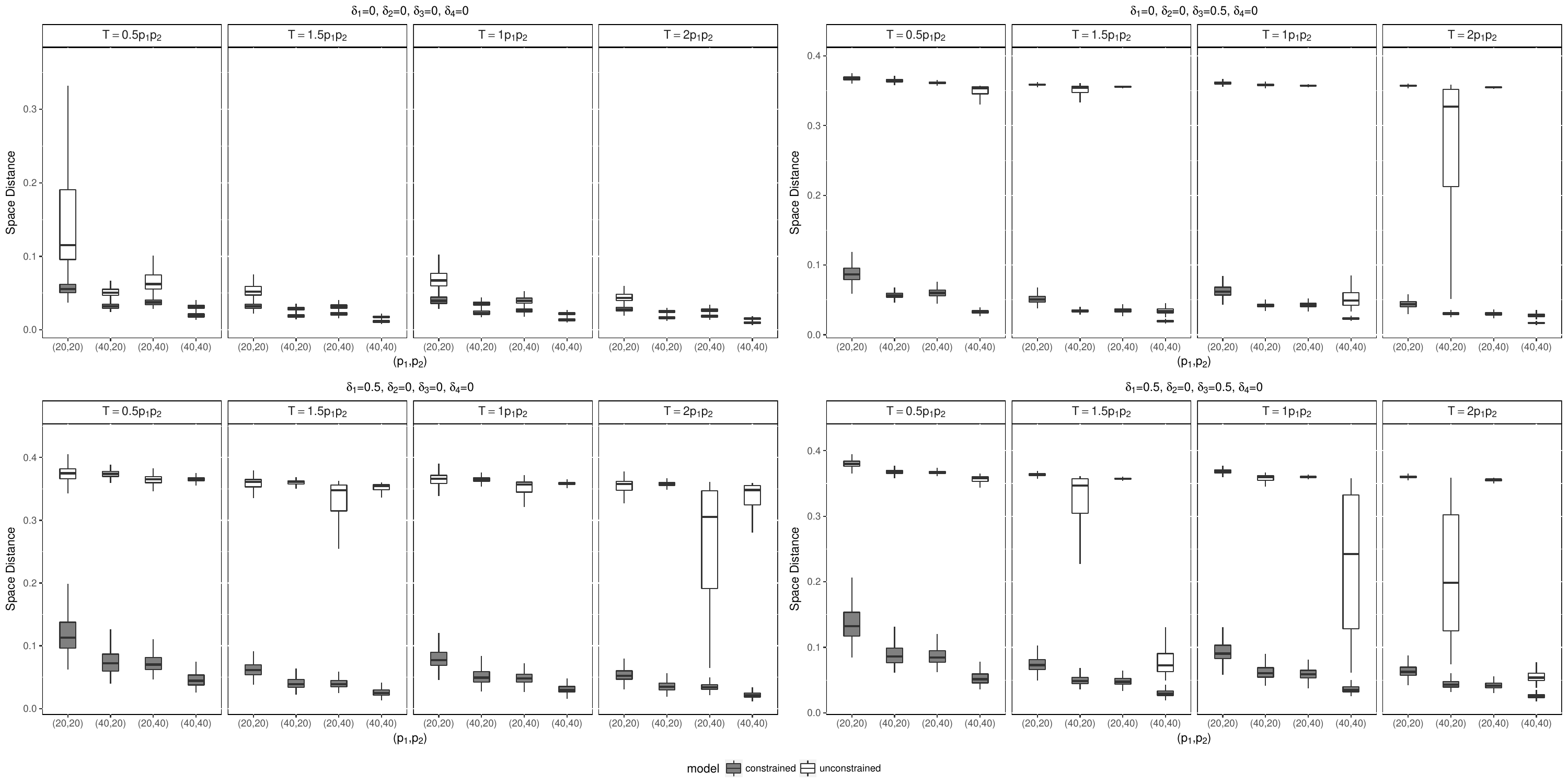}
    \caption{The strong factors case. Box-plots of the estimation accuracy measured by $\mathcal{D}(\widehat{\bQ}, \bQ)$ for partially constrained factor models.  The gray boxes are for the 
    constrained approach. The results are based on $500$ realizations. See Table~19  %\ref{table:partial_Q_mean_sd} 
    in Appendix~C % \ref{appendix:tables_simulation_results} 
    for the plotted values.}
    \label{fig:partial_Q_mean_sd_strong}
\end{figure}

\begin{figure}[ht!]
    \centering
    \includegraphics[width=\linewidth,height=\textheight,keepaspectratio=true]{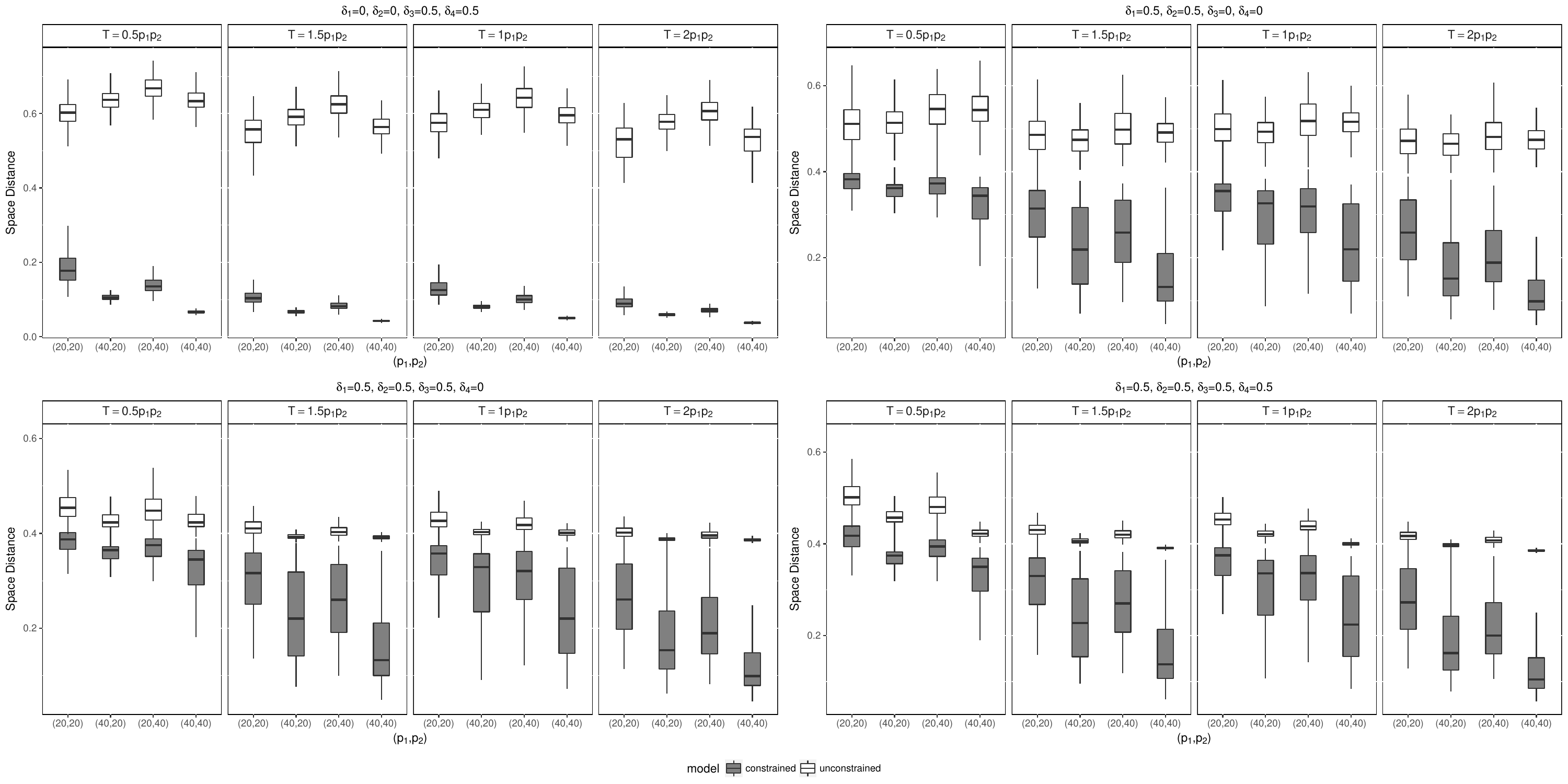}
    \caption{The weak factors case. Box-plots of the estimation accuracy measured by $\mathcal{D}(\widehat{\bQ}, \bQ)$ for partially constrained factor models. The gray boxes are 
    for the constrained approach. The results are based on $500$ realizations. See Table~19 % \ref{table:partial_Q_mean_sd} 
    in Appendix~C % \ref{appendix:tables_simulation_results} 
    for the plotted values.}
    \label{fig:partial_Q_mean_sd_weak}
\end{figure}

\section{Applications}
\label{sec:application}

In this section, we demonstrate the advantages of constrained matrix-variate factor models  
with three applications. In practice, the number of common factors ($k_1$, $k_2$) and the 
dimensions of constrained row and column loading spaces ($m_1$, $m_2$) must be pre-specified in order to determine an appropriate constrained factor model. The numbers of factors ($k_1$, $k_2$) can be determined by any existing methods, such as those in \cite{Lam-Yao-2012} and \cite{wang2018factor}. For any given ($k_1$, $k_2$), the dimension of constrained row and column loading spaces ($m_1$,$m_2$) can be determined by either $(a)$ prior or substantive knowledge or $(b)$ an empirical procedure. The results show that even simple grouping information can 
substantially increase the accuracy in estimation. 

\subsection{Example 1: Multinational Macroeconomic Indices}
\label{sec:application:subsec:macro_indices}

We apply the constrained and partially constrained factor models to the macroeconomic indexes collected from OECD. The dataset contains 10 quarterly macroeconomic indexes of 14 countries from 1990.Q2 to 2016.Q4  for 107 quarters. Thus, we have $T = 107$ and $p_1 \times p_2 = 14 \times 10$ matrix-valued time series. The countries include developed economies from North American, European, and Oceania. The indexes cover four major groups, namely production, consumer price, money market, and international trade.  Each original univariate time series is transformed by taking the first or second difference or logarithm to satisfy the mixing condition in \nameref{cond:Ft_alpha_mixing}. Detailed descriptions of the dataset and the transformation used are given in Table~15 % \ref{table:macro_index_data} 
and~16 %\ref{table:oecd_country_sel} 
of Appendix~B. %\ref{appendix:cross_country_macro_dataset}. 
Figure \ref{fig:oecd_mts_plot} shows the transformed time series of macroeconomic indicators of multiple countries. 

\begin{figure}[ht!]
    \centering
    \includegraphics[width=\linewidth,height=\textheight,keepaspectratio=true]{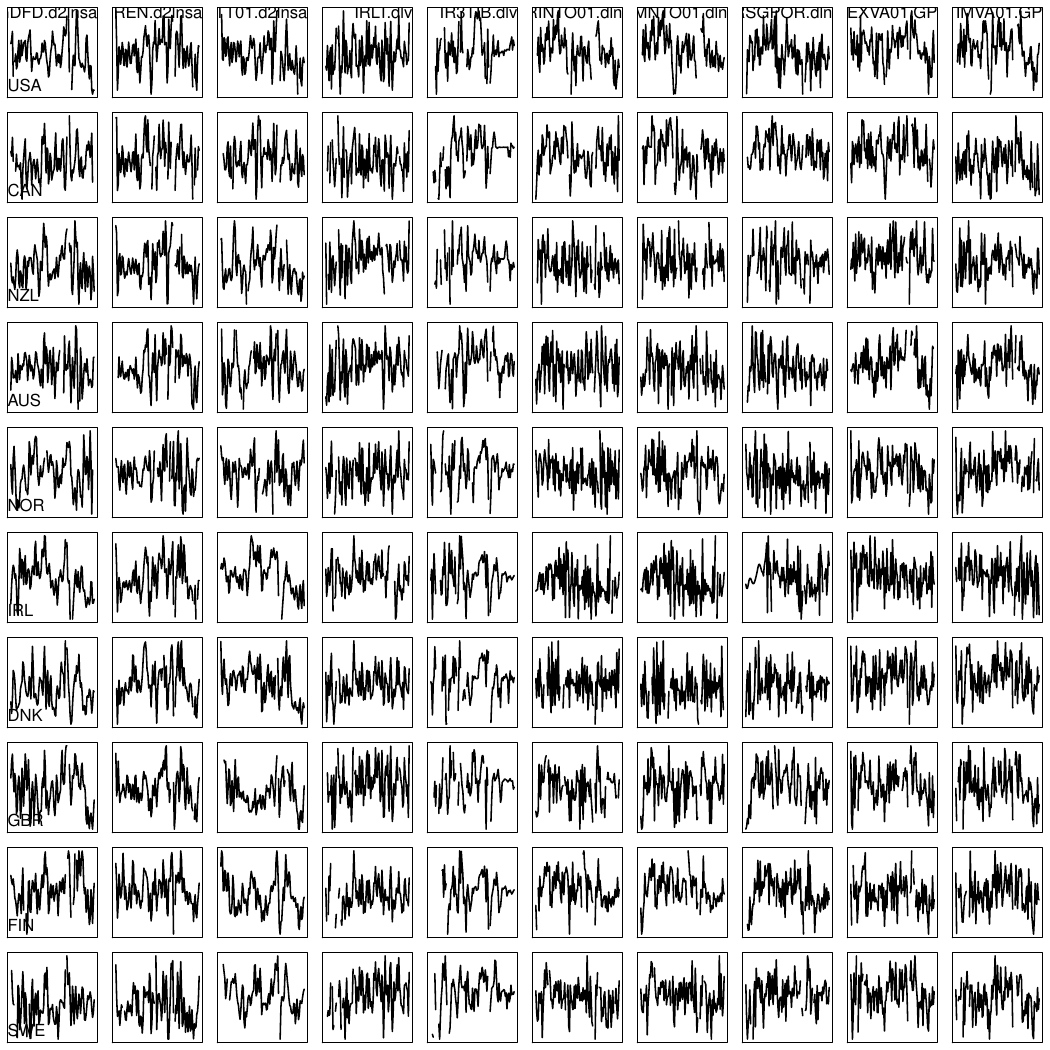}
    \caption{Time series plots of macroeconomic indicators of multiple countries (after data transformation). Only a subset of the coutries and idicators are ploted due to the space limit. Country and idicator abbreviation are given in Table~15 %\ref{table:macro_index_data} 
    and~16 %\ref{table:oecd_country_sel} 
    of Appendix C. %\ref{appendix:cross_country_macro_dataset}
    }
    \label{fig:oecd_mts_plot}
\end{figure}

We first fit an unconstrained matrix factor model and obtain estimates of the 
row loading matrix and the column loading matrix. In the row loading matrix, each row represents a country by its factor loadings, whereas, in the column loading matrix, each row represents a macroeconomic index by its factor loadings. A hierarchical clustering algorithm \citep{xu2005survey, murtagh2014ward} is employed to cluster countries and macroeconomic indices based on their representations in the common row and column factor spaces, respectively, under Euclidean distance and {\it ward.D} criterion. Figure \ref{fig:loading_matrix_clustering} shows the hierarchical clustering results. 
Based on the clustering result, we construct the row and column constraint matrices. It seems that the row constraint matrix divides countries into 6 groups: (i) United States and Canada; (ii) New Zealand and Australia; (iii) Norway; (iv) Ireland, Denmark, and United Kingdom; (v) Finland and Sweden; (vi) France, Netherlands, Austria, and Germany. The grouping more or less follows geographical partitions with Norway different from all others due to its rich oil production and other distinct economic characteristics. The column constraint matrix divides macroeconomic indexes into 5 categories: (i) GDP, production of total industry excluding construction, and production of total manufacturing ; (ii) long-term government bond yields and 3-month interbank rates and yields; (iii) total CPI and CPI of Food; (iv) CPI of Energy; (v) total exports value and total imports value in goods. Again, the grouping agrees with common economic knowledge. 

\begin{figure}[ht!]
\begin{subfigure}{0.5\textwidth}
\includegraphics[width=\linewidth,height=\textheight,keepaspectratio=true]{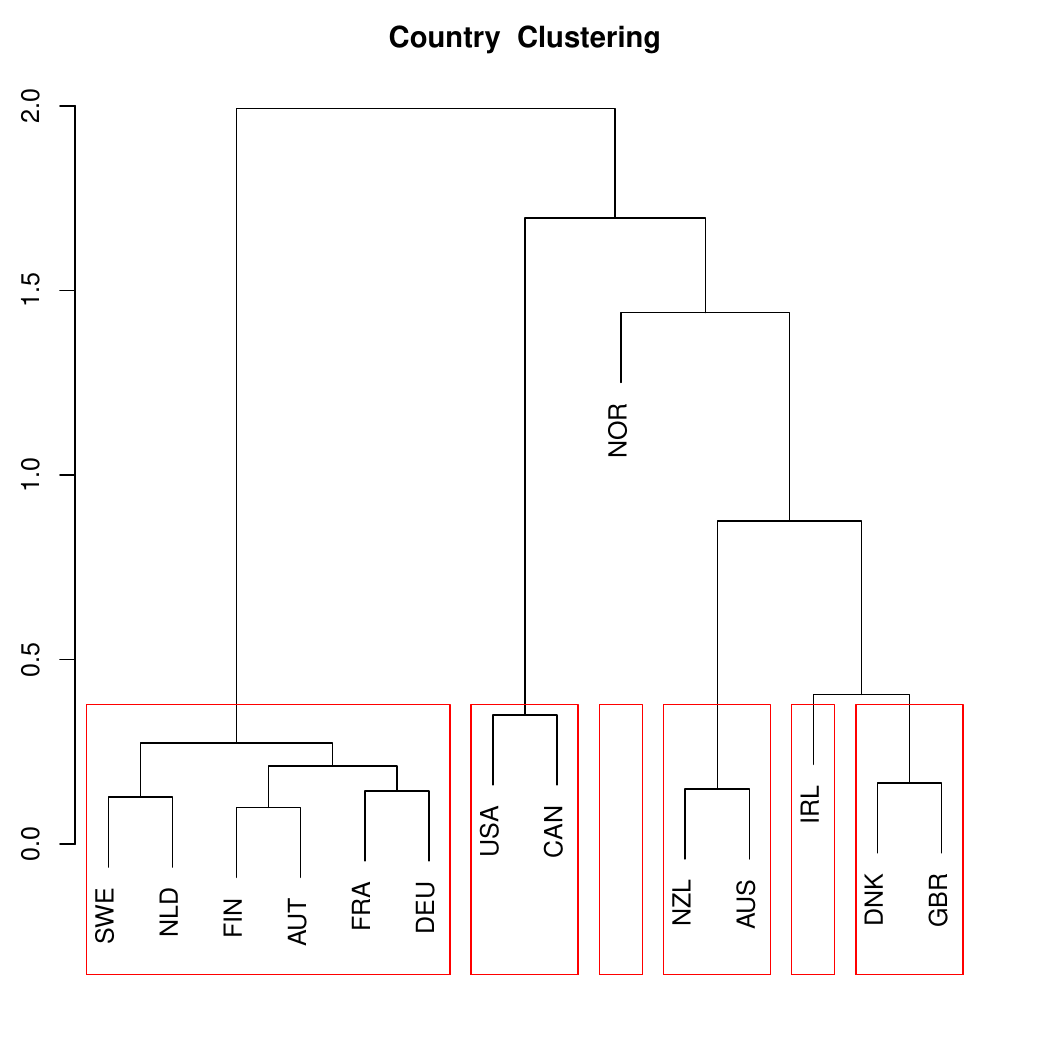} 
\caption{Country Loading Clustering}
\label{fig:country_clustering}
\end{subfigure}
\begin{subfigure}{0.5\textwidth}
\includegraphics[width=\linewidth,height=\textheight,keepaspectratio=true]{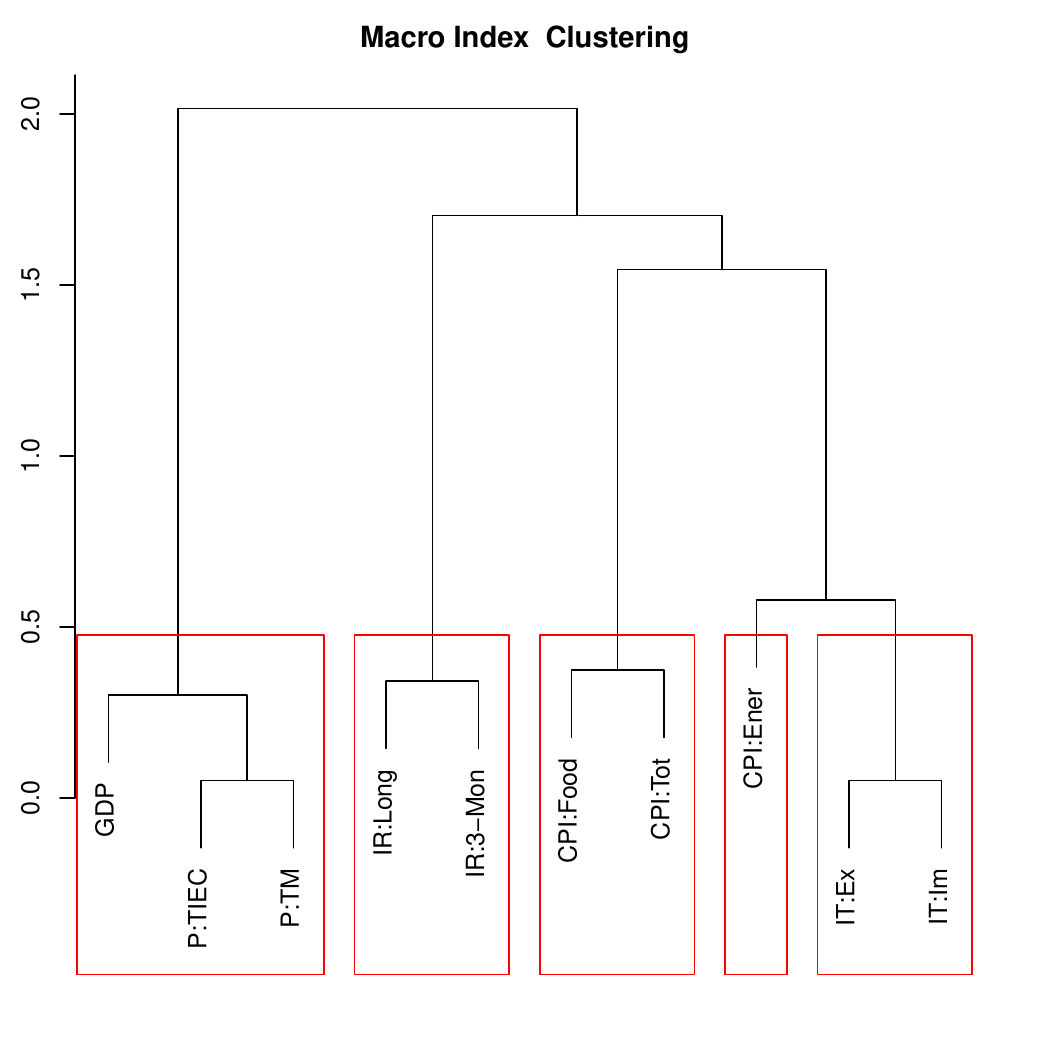}
\caption{Macroeconomic Index Loading Clustering}
\label{fig:macro_indices_clustering}
\end{subfigure}
 
\caption{Macroeconomic series: Clustering loading matrices}
\label{fig:loading_matrix_clustering}
\end{figure}

Table \ref{table:macro_Q_rot} shows estimates of the row and column loading matrices for  constrained and unconstrained $4 \times 4$ factor models. The loading matrices are normalized so that the norm of each column is one. They are also varimax-rotated to reveal a clear structure. The values shown are rounded values of the estimates multiplied by $10$ for ease in display. From the table, both the row and column loading matrices exhibit similar patterns between unconstrained and constrained models, partially validating the constraints while simplifying the analysis.

Table \ref{table:macro_Q} provides the estimates under the same setting as that of Table~\ref{table:macro_Q_rot} but without any rotation. 
%shows estimates of the row and column loading matrices for constrained and unconstrained $4 \times 4$ factor models. The loading matrices are normalized so that the norm of each column is one. The presented values are the estimated values multiplied by $10$ and rounded for a clear display. 
%From the table, both the row and column loading matrices exhibit similar patterns between unconstrained and constrained models, but that of the constrained model conveys more clearly the following observations: 
From the table, it is seen that except for the first common factor of the row loading matrices there exist 
some differences in the estimated loading matrices between unconstrained and constrained factor models. 
The results of constrained models convey more clearly the following observations. 
Consider the row factors. The first row common factor represents the status of global economy 
as it is a weighted average of all the countries under study. The remaining three row common factors mark  certain differences between country groups. For the column factors, the first column common factor is dominated by the price index and interest rates; The second column common factor is mainly the production and international trade; The remaining two column common factors represent interaction between price indexes, interest rates, productions, and international trade. 

\begin{table}[ht!]
	\centering
	\resizebox{\textwidth}{!}{%
		\begin{tabular}{ccccccccccccccccc}
			\hline
			Model & Loading & Row & USA & CAN & NZL & AUS & NOR & IRL & DNK & GBR & FIN & SWE & FRA & NLD & AUT & DEU \\ \hline
			\multirow{4}{*}{$R_{unc, rot}$} & \multirow{4}{*}{$\widehat{R}'_{rot}$} & 1 & \cellcolor[HTML]{EFEFEF}7 & \cellcolor[HTML]{EFEFEF}7 & 1 & 1 & -1 & \cellcolor[HTML]{EFEFEF}-2 & \cellcolor[HTML]{EFEFEF}-1 & \cellcolor[HTML]{EFEFEF}0 & 1 & 0 & 0 & 0 & 0 & -1 \\
			&  & 2 & 0 & 1 & -2 & -1 & 1 & \cellcolor[HTML]{EFEFEF}1 & \cellcolor[HTML]{EFEFEF}1 & \cellcolor[HTML]{EFEFEF}2 & \cellcolor[HTML]{EFEFEF}4 & \cellcolor[HTML]{EFEFEF}3 & \cellcolor[HTML]{EFEFEF}4 & \cellcolor[HTML]{EFEFEF}4 & \cellcolor[HTML]{EFEFEF}4 & \cellcolor[HTML]{EFEFEF}4 \\
			&  & 3 & 2 & -1 & \cellcolor[HTML]{EFEFEF}5 & \cellcolor[HTML]{EFEFEF}5 & 1 & \cellcolor[HTML]{EFEFEF}5 & \cellcolor[HTML]{EFEFEF}3 & \cellcolor[HTML]{EFEFEF}2 & -1 & 1 & 1 & 0 & 0 & 0 \\
			&  & 4 & -1 & 1 & 1 & 2 & \cellcolor[HTML]{EFEFEF}9 & -3 & 0 & 0 & 0 & 1 & -1 & 1 & 0 & 0 \\ \cline{3-17} 
			\multirow{4}{*}{$R_{con, rot}$} & \multirow{4}{*}{$\widehat{R}_{rot}' H_R'$} & 1 & \cellcolor[HTML]{EFEFEF}6 & \cellcolor[HTML]{EFEFEF}6 & 0 & 0 & 0 & \cellcolor[HTML]{EFEFEF}2 & \cellcolor[HTML]{EFEFEF}2 & \cellcolor[HTML]{EFEFEF}2 & -1 & -1 & 0 & 0 & 0 & 0 \\
			&  & 2 & -1 & -1 & 0 & 0 & 0 & \cellcolor[HTML]{EFEFEF}3 & \cellcolor[HTML]{EFEFEF}3 & \cellcolor[HTML]{EFEFEF}3 & \cellcolor[HTML]{EFEFEF}4 & \cellcolor[HTML]{EFEFEF}4 & \cellcolor[HTML]{EFEFEF}3 & \cellcolor[HTML]{EFEFEF}3 & \cellcolor[HTML]{EFEFEF}3 & \cellcolor[HTML]{EFEFEF}3 \\
			&  & 3 & 0 & 0 & \cellcolor[HTML]{EFEFEF}7 & \cellcolor[HTML]{EFEFEF}7 & 0 & \cellcolor[HTML]{EFEFEF}1 & \cellcolor[HTML]{EFEFEF}1 & \cellcolor[HTML]{EFEFEF}1 & 1 & 1 & -1 & -1 & -1 & -1 \\
			&  & 4 & 0 & 0 & 0 & 0 & \cellcolor[HTML]{EFEFEF}10 & 0 & 0 & 0 & 1 & 1 & 0 & 0 & 0 & 0 \\ \hline
		\end{tabular}%
	} \\
	
	\vspace{0.2in}
	
	\resizebox{\textwidth}{!}{%
		\begin{tabular}{ccccccccccccc}
			\hline
			Model & Loading & Row & CPI:Food & CPI:Tot & CPI:Ener & IR:Long & IR:3-Mon & P:TIEC & P:TM & GDP & IT:Ex & IT:Im \\ \hline
			\multirow{4}{*}{$C_{unc, rot}$} & \multirow{4}{*}{$\widehat{C}'_{rot}$} & 1 & \cellcolor[HTML]{EFEFEF}6 & \cellcolor[HTML]{EFEFEF}7 & 3 & -1 & 1 & 0 & 0 & -1 & -1 & 0 \\
			&  & 2 & -2 & 1 & \cellcolor[HTML]{EFEFEF}4 & 1 & -1 & 0 & 0 & 0 & \cellcolor[HTML]{EFEFEF}6 & \cellcolor[HTML]{EFEFEF}6 \\
			&  & 3 & 0 & 0 & 1 & \cellcolor[HTML]{EFEFEF}8 & \cellcolor[HTML]{EFEFEF}6 & -1 & 0 & 1 & 0 & 0 \\
			&  & 4 & 1 & -1 & 0 & 0 & 0 & \cellcolor[HTML]{EFEFEF}6 & \cellcolor[HTML]{EFEFEF}6 & \cellcolor[HTML]{EFEFEF}5 & 0 & 0 \\ \cline{3-13} 
			\multirow{4}{*}{$C_{con, rot}$} & \multirow{4}{*}{$\widehat{C}'_{rot} H_C'$} & 1 & \cellcolor[HTML]{EFEFEF}7 & \cellcolor[HTML]{EFEFEF}7 & 0 & 0 & 0 & 0 & 0 & 0 & 0 & 0 \\
			&  & 2 & 0 & 0 & \cellcolor[HTML]{EFEFEF}6 & 0 & 0 & 0 & 0 & 0 & \cellcolor[HTML]{EFEFEF}6 & \cellcolor[HTML]{EFEFEF}6 \\
			&  & 3 & 0 & 0 & 0 & \cellcolor[HTML]{EFEFEF}7 & \cellcolor[HTML]{EFEFEF}7 & 0 & 0 & 0 & 0 & 0 \\
			&  & 4 & 0 & 0 & -2 & 0 & 0 & \cellcolor[HTML]{EFEFEF}6 & \cellcolor[HTML]{EFEFEF}6 & \cellcolor[HTML]{EFEFEF}6 & 1 & 1 \\ \hline
		\end{tabular}%
	}
	\caption{Estimations of row and column loading matrices (varimax rotated) of constrained and unconstrained matrix factor models for multinational macroeconomic indices. The loadings matrix are multiplied by $10$ and rounded to integers for ease in display.}
	\label{table:macro_Q_rot}
\end{table}

\begin{table}[ht!]
	\centering
	\resizebox{\textwidth}{!}{%
		\begin{tabular}{ccccccccccccccccc}
			\hline
			Model & Loading & Row & USA & CAN & NZL & AUS & NOR & IRL & DNK & GBR & FIN & SWE & FRA & NLD & AUT & DEU \\ \hline
			\multirow{4}{*}{$R_{unc}$} & \multirow{4}{*}{$\widehat{R}'$} & 1 & 3 & 2 & 2 & 2 & 2 & 2 & 3 & 3 & 3 & 3 & 3 & 3 & 3 & 3 \\
			&  & 2 & 4 & 2 & 5 & 5 & 1 & 0 & 1 & 0 & -3 & -1 & -2 & -2 & -2 & -3 \\
			&  & 3 & 3 & 6 & -2 & -2 & 4 & -5 & -3 & -1 & 1 & 0 & -1 & 1 & 0 & 0 \\
			&  & 4 & -4 & -3 & 0 & 2 & 8 & -1 & 1 & 0 & -1 & 1 & 0 & 1 & 0 & 0 \\ \cline{3-17} 
			\multirow{4}{*}{$R_{con}$} & \multirow{4}{*}{$\widehat{R}' H_R'$} & 1 & 1 & 1 & 2 & 2 & 2 & 3 & 3 & 3 & 4 & 4 & 3 & 3 & 3 & 3 \\
			&  & 2 & 5 & 5 & 3 & 3 & 4 & 0 & 0 & 0 & -2 & -2 & -2 & -2 & -2 & -2 \\
			&  & 3 & -1 & -1 & 5 & 5 & -6 & 0 & 0 & 0 & 0 & 0 & -1 & -1 & -1 & -1 \\
			&  & 4 & -4 & -4 & 3 & 3 & 6 & -2 & -2 & -2 & 1 & 1 & -1 & -1 & -1 & -1 \\ \hline
		\end{tabular}%
	} \\
	
	\vspace{0.2in}
	
	\resizebox{\textwidth}{!}{%
		\begin{tabular}{ccccccccccccc}
			\hline
			Model & Loading & Row & CPI:Food & CPI:Ener & CPI:Tot & IR:Long & IR:3-Mon & P:TIEC & P:TM & GDP & IT:Ex & IT:Im \\ \hline
			\multirow{4}{*}{$C_{unc}$} & \multirow{4}{*}{$\widehat{C}'$} & 1 & 1 & 4 & 2 & 4 & 3 & 3 & 3 & 3 & 4 & 4 \\
			&  & 2 & 5 & 3 & 6 & -1 & 1 & -3 & -4 & -4 & 0 & 0 \\
			&  & 3 & 5 & -1 & 2 & -1 & 1 & 4 & 4 & 3 & -4 & -4 \\
			&  & 4 & 0 & -1 & -2 & 7 & 5 & -2 & -2 & 0 & -3 & -3 \\ \cline{3-13} 
			\multirow{4}{*}{$C_{con}$} & \multirow{4}{*}{$\widehat{C}' H_C'$} & 1 & 6 & -2 & 6 & 4 & 4 & 0 & 0 & 0 & -2 & -2 \\
			&  & 2 & 0 & 0 & 0 & 3 & 3 & 5 & 5 & 5 & 3 & 3 \\
			&  & 3 & -3 & 3 & -3 & 5 & 5 & -3 & -3 & -3 & 1 & 1 \\
			&  & 4 & 3 & 5 & 3 & -1 & -1 & -2 & -2 & -2 & 5 & 5 \\ \hline
		\end{tabular}%
	}
	\caption{Estimations of row and column loading matrices of constrained and unconstrained matrix factor models for multinational macroeconomic indices. No rotation is used. 
	The loadings matrix are multiplied by $10$ and rounded 
	to integers for ease in display.}
	\label{table:macro_Q}
\end{table}

Table \ref{table:macro_indices_ss_outofsample_10CV} compares the out-of-sample performance of unconstrained, constrained, and partially constrained factor models using a 10-fold cross validation (CV) for models with different number of factors.
We divide the entire time span into 10 sections and choose each of them as testing data. With time series data, the training data may contain two disconnected time spans. We calculate two matrices $\widehat{\bM}_1$ and $\widehat{\bM}_2$ according to (\ref{eqn:Mhat_def}) for the two disconnected time spans separately. The matrix $\widehat{\bM}$ is re-defined as the sum of $\widehat{\bM}_1$ and $\widehat{\bM}_2$. Loading matrices and latent dimensions are estimated from this newly defined $\widehat{\bM}$ with procedures in Section \ref{sec:estimation}.  
Residual sum of squares (RSS), their ratios to the total sum of squares (RSS/TSS), and the number of parameters are the average of the 10-fold CV. Clearly, the constrained factor model uses far fewer parameters in the loading matrices yet achieves slightly better results than the   
unconstrained model. Using the same number of parameters, the partially constrained model
is able to reduce markedly the RSS over the unconstrained model. 

In this particular application, the constrained matrix factor model with the specified constraint matrices seems appropriate and plausible. If incorrect structures (constraint matrices) are imposed on the model, then the constrained model may become inappropriate. As we can see from the next example, a single orthogonal constraint actually hurts the performance. In cases like this, we need a second or a third constraint to achieve satisfactory performance. Nevertheless, the results from the constrained model are better than those from the unconstrained model.  

% \usepackage{graphicx}
\begin{table}[ht!]
\centering
\resizebox{0.7\textwidth}{!}{%
\begin{tabular}{cccccc}
	\hline
	\rowcolor[HTML]{EFEFEF} 
	Model & \# Factor 1 & \# Factor 2 & RSS & RSS/TSS & \# Parameters \\ \hline
	Full & (6,5) &  & 570.50 & 0.449 & 134 \\
	Constrained & (6,5) &  & 560.31 & 0.442 & 61 \\
	\rowcolor[HTML]{EFEFEF} 
	Partial & (6,5) & (6,5) & 454.41 & 0.358 & 134 \\ \hline
	Full & (5,5) &  & 613.26 & 0.482 & 120 \\
	Constrained & (5,5) &  & 604.63 & 0.477 & 55 \\
	\rowcolor[HTML]{EFEFEF} 
	Partial & (5,5) & (5,5) & 516.27 & 0.407 & 120 \\ \hline
	Full & (4,5) &  & 658.15 & 0.517 & 106 \\
	Constrained & (4,5) &  & 649.85 & 0.512 & 49 \\
	\rowcolor[HTML]{EFEFEF} 
	Partial & (4,5) & (4,5) & 576.94 & 0.454 & 106 \\ \hline
	Full & (4,4) &  & 729.46 & 0.573 & 96 \\
	Constrained & (4,4) &  & 721.96 & 0.568 & 44 \\
	\rowcolor[HTML]{EFEFEF} 
	Partial & (4,4) & (4,4) & 657.13 & 0.517 & 96 \\ \hline
	Full & (3,4) &  & 787.80 & 0.620 & 82 \\
	Constrained & (3,4) &  & 768.64 & 0.605 & 38 \\
	\rowcolor[HTML]{EFEFEF} 
	Partial & (3,4) & (3,4) & 719.46 & 0.567 & 82 \\ \hline
	Full & (3,3) &  & 868.43 & 0.684 & 72 \\
	Constrained & (3,3) &  & 852.76 & 0.671 & 33 \\
	\rowcolor[HTML]{EFEFEF} 
	Partial & (3,3) & (3,3) & 813.16 & 0.640 & 72 \\ \hline
\end{tabular}%
}
\caption{Results of 10-fold CV of out-of-sample performance for the multinational macroeconomic 
indexes. The numbers shown are average over the cross validation, where RSS and TSS stand for 
residual and total sum of squares, respectively.}
\label{table:macro_indices_ss_outofsample_10CV}
\end{table}

\subsection{Example 2: Company Financial Measurements}
\label{sec:application:subsec:Company Financials}

In this application, we investigate the constrained matrix-variate factor models for the time series of 16 quarterly financial measurements of 200 companies from 2006.Q1 to 2015.Q4 for 40 observations. Appendix~D %\ref{appendix:copfin_data} 
contains the descriptions of variables used and their definitions, the 200 companies and their corresponding industry group and sector information. 
Data are arranged in matrix-variate time series format. At each $t$, we observe a $16 \times 200$ matrix, whose rows represent financial variables and columns represent companies. Thus we have $T=40$, $p_1=16$ and $p_2=200$. The total number of time series is $3,200$. Following the convention in eigenanalysis, we standardize the individual series before applying factor analysis.  
This data set was used in \cite{wang2018factor} for an unconstrained matrix factor model. 

The column constraint matrix $\bH_C$ is constructed based on the industrial 
classification of Bloomberg. 
The $200$ companies are classified into $51$ industrial groups, such as biotechnology, oil \& gas, computer,  among others. Thus the dimension of $\bH_C$ is $200 \times 51$. Since we do not have adequate prior knowledge on corporate financial, we do not impose any constraint on the row loading matrix. Thus, in this application, we use $\bH_R = \bI_{16}$. 

We apply the unconstrained model (\ref{eqn:mfm}), the orthogonal constrained model (\ref{eqn:cmfm_trans_orth}), and the partial constrained model (\ref{eqn:pcmfm_explicit}) 
to the data.
Table \ref{table:copfin_ss_outofsample_10CV} shows the average residual sum of squares (RSS) and their ratios to the total sum of squares (TSS) from a 10-fold CV for models with different number of factors. 
Again, it is clear, from the table, that the constrained matrix factor models use fewer numbers of parameters in loading matrices and achieve similar results. If we use the same number of parameters in the loading matrices, variances explained by the constrained matrix factor models are much larger than those of the unconstrained ones, indicating the impact of over-parameterization.  
This application with $3,200$ time series is typical in high-dimensional time series analysis. 
The number of parameters involved is usually huge in a unconstrained model. Via the example, we showed that constrained matrix factor models can substantially reduce the number of parameters while keep the same explanation power.  

% \usepackage{multirow}
\begin{table}[ht!]
\centering
\resizebox{0.7\textwidth}{!}{%
\begin{tabular}{cccccc}
\hline
Model                                             & \# Factor 1 & \# Factor 2 & RSS     & RSS/SST & \# parameters \\ \hline
\rowcolor[HTML]{EFEFEF}
                                                  & (4,10)      &             & 8140.32 & 0.869   & 2064          \\
                                                  & (4,12)      &             & 7990.04 & 0.853   & 2464          \\
\multirow{-3}{*}{Full}                            & (4,19)      &             & 7587.11 & 0.810   & 3864          \\
\rowcolor[HTML]{EFEFEF}
Constrained                                       & (4,10)      &             & 8062.63 & 0.861   & 574           \\ 
                                                  & (4,10)      & (4,2)       & 7969.83 & 0.851   & 936          \\ 
\multirow{-2}{*}{Partial}                         & (4,10)      & (4,9)       & 7623.25 & 0.814   & 1979         \\ \hline
\rowcolor[HTML]{EFEFEF}
                                                  & (4, 20)     &             & 7539.68 & 0.805   & 4064          \\
                                                  & (4, 27)     &             & 7261.49 & 0.775   & 5464          \\
\multirow{-3}{*}{Full}                            & (4, 39)     &             & 6872.18 & 0.734   & 7864          \\
\rowcolor[HTML]{EFEFEF}
Constrained                                       & (4, 20)     &             & 7646.70 & 0.816   & 1084          \\ 
                                                  & (4, 20)     & (4,7)       & 7292.06 & 0.779   & 2191          \\
\multirow{-2}{*}{Partial}                         & (4, 20)     & (4,19)      & 6815.96 & 0.728   & 3979         \\ \hline
\rowcolor[HTML]{EFEFEF}
                                                  & (5,10)      &             & 8012.10 & 0.855   & 2080          \\
                                                  & (5,12)      &             & 7849.34 & 0.838   & 2480          \\
\multirow{-3}{*}{Full}                            & (5,19)      &             & 7420.04 & 0.792   & 3880          \\
\rowcolor[HTML]{EFEFEF}
Constrained                                       & (5,10)      &             & 7942.95 & 0.848   & 590           \\
                                                  & (5,10)      & (5,2)       & 7849.40 & 0.838   & 968          \\ 
\multirow{-2}{*}{Partial}                         & (5,10)      & (5,9)       & 7472.10 & 0.798   & 2011          \\ \hline
\rowcolor[HTML]{EFEFEF}
                                                  & (5,20)      &             & 7368.63 & 0.787   & 7960          \\
                                                  & (5,23)      &             & 7250.73 & 0.774   & 4680          \\
\multirow{-3}{*}{Full}                            & (5,39)      &             & 6641.13 & 0.709   & 7880          \\
\rowcolor[HTML]{EFEFEF}
Constrained                                       & (5,20)      &             & 7489.20 & 0.800   & 1100          \\ 
                                                  & (5,20)      & (5,3)       & 7357.80 & 0.786   & 1627          \\
\multirow{-2}{*}{Partial}                         & (5,20)      & (5,19)      & 6595.03 & 0.704   & 4011          \\ \hline
\rowcolor[HTML]{EFEFEF}
                                                  & (5,30)      &             & 6960.70 & 0.743   & 6080          \\
                                                  & (5,34)      &             & 6813.93 & 0.727   & 6880          \\
\multirow{-3}{*}{Full}                            & (5,59)      &             & 5988.15 & 0.639   & 11880         \\
\rowcolor[HTML]{EFEFEF}
Constrained                                       & (5,30)      &             & 7184.53 & 0.767   & 1610          \\
                                                  & (5,30)      & (5,4)       & 6997.21 & 0.747   & 2286          \\ 
\multirow{-2}{*}{Partial}                         & (5,30)      & (5,29)      & 5936.64 & 0.634   & 6011          \\ \hline
\end{tabular} }
\caption{Summary of 10-fold CV of out-of-sample analysis for the 16 corporate financial measurements for each of 200 companies. 
The numbers shown are average over the cross validation and 
RSS and TSS denote, respectively, the residual and total sum of squares.}
\label{table:copfin_ss_outofsample_10CV}
\end{table}

\subsection{Example 3: Fama-French 10 by 10 Series}
\label{sec:application:subsec:Fama-French}

Finally, we investigate constrained matrix-variate factor models for the monthly market-adjusted return series of Fama-French $10 \times 10$ portfolios from January 1964 to December 2015 for 
624 months and overall $62,400$ observations. The portfolios are the intersections of 10 portfolios formed by size (market equity, ME) and 10 portfolios formed by the ratio of book equity to market equity (BE/ME). Thus, we have $T=624$ and $p_1 \times p_2 = 10 \times 10$ matrix-variate 
time series. The series are constructed by subtracting the monthly excess market returns from each of the original portfolio returns obtained from \cite{FFdata}, 
so they are free of the market impact. 

Using an unconstrained matrix factor model, \cite{wang2018factor} carried out 
a clustering analysis on the ME and BE/ME loading matrices after rotation. Their results suggest $\bH_R = [\utwi{h}_{R_1}, \utwi{h}_{R_2}, \utwi{h}_{R_3}]$, where $\utwi{h}_{R_1} = [\mathbf{1}(5)/\sqrt{5}, \mathbf{0}(5)]$, $\utwi{h}_{R_2} = [\mathbf{0}(5), \mathbf{1}(4)/2, 0]$, and $\utwi{h}_{R_3} = [\mathbf{0}(9), 1]$. Therefore, ME factors are classified into three groups of smallest $5$ ME's, middle $4$ ME's, and the largest ME, respectively. For cases when we need 4 row constraints, we redefine $\utwi{h}_{R_2}=[\mathbf{0}(5), \mathbf{1}(3)/\sqrt{3}, \mathbf{0}(2)]$ and add a fourth column $\utwi{h}_{R_4} = [\mathbf{0}(8),1,0]$. For column constraints, $\bH_C = [\utwi{h}_{C_1}, \utwi{h}_{C_2}, \utwi{h}_{C_3}]$, where $\utwi{h}_{C_1} = [1, \mathbf{0}(9)]$, $\utwi{h}_{C_2} = [0, \mathbf{1}(3)/\sqrt{3}, \mathbf{0}(6)]$, $\utwi{h}_{C_3} = [\mathbf{0}(4), \mathbf{1}(6)]$. 
Therefore, BE/ME factors are divided into three groups of the smallest BE/ME's, middle $3$ BE/ME's, and the $6$ largest BE/ME, respectively. For cases when we need 4 column constraints, we redefine $\utwi{h}_{C_3} = [\mathbf{0}(4), \mathbf{1}(4)/2, \mathbf{0}(2)]$ and add a fourth column $\utwi{h}_{C_4} = [\mathbf{0}(8), \mathbf{1}(2)]$.

Table \ref{table:famafrench_Q_rot} shows the estimates of the loading matrices for the constrained and unconstrained $2 \times 2$ factor models. The loading matrices are varimax-rotated for ease in  interpretation and normalized so that the norm of each column is one. From the table, the loading matrices exhibit similar patterns, but those of the constrained model convey the following observations more clearly. Consider the row factors. The first factor represents the difference between the average of the $5$ smallest ME group and the weighted average of the remaining portfolio whereas the second factor is mainly the average of the medium 
4 ME portfolios. For the column loading matrix, the first factor is a weighted average of 
the smallest BE/ME portfolio and the middle three portfolios. 
The second factor marks the difference between the smallest BE/ME portfolio 
from a weighted average of the two remaining groups. 
Finally, it is interesting to see that the constrained model uses only 16 parameters, yet 
it can reveal information similar to the unconstrained model that employs 
40 parameters. This result demonstrates the power of using constrained factor models. 

%Another difference between the constrained and unconstrained factor models is that in this particular instance the constrained model uses $16$ parameters whereas the unconstrained one has $40$ parameters in the loading matrices. 

% \usepackage{multirow}
% \usepackage{graphicx}
\begin{table}[ht!]
	\centering
	\resizebox{\textwidth}{!}{%
		\begin{tabular}{ccccccccccccc}
			\hline
			Model & Loading & Column & \multicolumn{10}{c}{Rotated Estimated Loadings} \\ \hline
			\multirow{4}{*}{$R_{u}$} & \multirow{2}{*}{$\widehat{R}'$} & 1 & 0.43 & 0.46 & 0.44 & 0.43 & 0.33 & 0.16 & 0.05 & -0.02 & -0.20 & -0.23 \\
			&  & 2 & -0.01 & -0.01 & -0.05 & 0.09 & 0.18 & 0.39 & 0.39 & 0.62 & 0.51 & 0.16 \\ \cline{2-13} 
			& \multirow{2}{*}{$\widehat{R}' H_R'$} & 1 & 0.44 & 0.44 & 0.44 & 0.44 & 0.44 & -0.04 & -0.04 & -0.04 & -0.04 & -0.15 \\
			&  & 2 & 0.04 & 0.04 & 0.04 & 0.04 & 0.04 & 0.50 & 0.50 & 0.50 & 0.50 & 0.06 \\ \hline 
			\multirow{4}{*}{$C_{u}$} & \multirow{2}{*}{$\widehat{C}'$} & 1 & 0.70 & 0.48 & 0.37 & 0.30 & 0.14 & 0.07 & 0.05 & -0.05 & -0.09 & 0.15 \\
			&  & 2 & 0.29 & -0.07 & -0.10 & -0.23 & -0.30 & -0.32 & -0.34 & -0.44 & -0.48 & -0.34 \\ \cline{2-13} 
			& \multirow{2}{*}{$\widehat{C}' H_C' $} & 1 & 0.78 & 0.36 & 0.36 & 0.36 & 0 & 0 & 0 & 0 & 0 & 0 \\
			&  & 2 & 0.24 & -0.18 & -0.18 & -0.18 & -0.37 & -0.37 & -0.37 & -0.37 & -0.37 & -0.37 \\ \hline
		\end{tabular}%
	}
	\caption{Estimates of the loading matrices of constrained and unconstrained matrix factor modes for Fama-French $10 \times 10$ portfolio returns. The loading matrices are varimax rotated and normalized for ease in comparison.}
	\label{table:famafrench_Q_rot}
\end{table}

%\textcolor{blue}{
Table \ref{table:famafrench_ss_outofsample_10CV} compares the out-of-sample performance of unconstrained and constrained matrix factor models using a 10-fold CV for models with different number of factors constructed similarly to that of Table \ref{table:macro_indices_ss_outofsample_10CV}. In this case, the prediction RSS of the constrained model is slightly larger than that of the unconstrained one with the same number of factors, which may results from the misspecification of the constrained matrices. Testing the adequacy of the constrained matrix is an important research topic to be addressed in future research. On the other hand, the constrained model uses 
a much smaller number of parameters than the unconstrained model.  
% \usepackage{multirow}
% \usepackage{graphicx}
\begin{table}[ht!]
	\centering
	\resizebox{0.7\textwidth}{!}{%
		\begin{tabular}{cccccc}
			\hline 
			Model & \# Factor 1 & \# Factor 2 & RSS & RSS/SST & \# Parameters \\ \hline
            \rowcolor[HTML]{EFEFEF}
			 & (3,3) &  & 3064.40 & 0.500 & 60 \\
             & (3,4) &  & 2905.79 & 0.474 & 70 \\
            \multirow{-3}{*}{Full} & (3,6) &  & 2644.59 & 0.431 & 90 \\
            \rowcolor[HTML]{EFEFEF}
			Constrained & (3,3) &  & 3115.16 & 0.508 & 24 \\			 
			& (3,3) & (3,3) & 2819.06 & 0.460 & 60 \\
			\multirow{-2}{*}{Partial} & (3,3) & (1,1) & 3079.79 & 0.502 & 36 \\ \hline
            \rowcolor[HTML]{EFEFEF}
			& (3,2) &  & 3316.55 & 0.541 & 50 \\
            \multirow{-2}{*}{Full}& (3,4) &  & 2905.79 & 0.474 & 70 \\
            \rowcolor[HTML]{EFEFEF}
			Constrained & (3,2) &  & 3361.03 & 0.548 & 18 \\ 
			& (3,2) & (3,2) & 3169.79 & 0.517 & 50 \\ 
			\multirow{-2}{*}{Partial} & (3,2) & (1,1) & 3323.25 & 0.542 & 31 \\ \hline
            \rowcolor[HTML]{EFEFEF}
			& (2,3) &  & 3269.50 & 0.533 & 50 \\
            & (2,4) &  & 3152.63 & 0.514 & 60 \\
            \multirow{-3}{*}{Full}& (2,6) &  & 2976.18 & 0.431 & 90 \\
            \rowcolor[HTML]{EFEFEF}
			Constrained & (2,3) &  & 3372.79 & 0.550 & 18 \\ 
			& (2,3) & (2,3) & 3154.36 & 0.514 & 50 \\ 
			\multirow{-2}{*}{Partial} & (2,3) & (1,2) & 3296.73 & 0.538 & 37 \\ \hline
            \rowcolor[HTML]{EFEFEF}
			& (2,2) &  & 3473.32 & 0.567 & 40 \\
            & (2,3) &  & 3269.50 & 0.533 & 50 \\
            \multirow{-3}{*}{Full} & (2,4) &  & 3152.63 & 0.514 & 60 \\
            \rowcolor[HTML]{EFEFEF}
			Constrained & (2,2) &  & 3535.56 & 0.577 & 16 \\
			& (2,2) & (2,2) & 3415.25 & 0.557 & 40 \\ 
			\multirow{-2}{*}{Partial} & (2,2) & (2,1) & 3486.15 & 0.569 & 33 \\ \hline
		\end{tabular}%
	}
\caption{Performance of out-of-sample  10-fold CV of  constrained and unconstrained factor 
models using Fama-French $10 \times 10$ portfolio return series, where $RSS$ and $RSS/TSS$ 
denote, respectively, the residual and total sum of squares.}
\label{table:famafrench_ss_outofsample_10CV} 
\end{table}

\section{Summary and Discussion} \label{sec:summary}

This paper established a general framework for incorporating domain or prior knowledge induced linear constraints in the matrix factor model. We developed efficient estimation procedures for
constrained, multi-term, and partially constrained matrix factor models. Constraints can be used to achieve parsimony in parameterization, to facilitate factor interpretation, and to target specific factors indicated by the domain theories. We derived  asymptotic theorems justifying the benefits of imposing constraints. Simulation results confirmed
the advantages of constrained matrix factor model over the unconstrained one in finite samples.
Finally, we illustrated 
the applications of constrained matrix factor models with three real data sets, 
where the constrained  factor models outperform their unconstrained counterparts in explaining the 
variabilities of the data using out-of-sample  $10$-fold cross validation and in factor interpretation. 

Under the model setting we adopt, both strong and weak factors exist in the dynamic component. The proposed constrained model incorporates prior information and improves the rates of convergence in the case of weak factors. For the strong factor case, it achieves the same asymptotic rates as those of the unconstrained models. Yet it entails smaller number of parameters and requires weaker assumption on the growth rates of dimensions and sample size. Several interesting topics are open for further researches. Firstly, a natural question is how we know the existence of weak factors in real data. \cite{Lam-Yao-2012} utilized a two-step approach to facilitate the discovery of weak factors that may be masked by strong factors. They run a second decomposition to the residual from the first step to find the weak factors that may be masked from the strong factor in the first step. A possible method to test the existence of weak factor will be to test the existence of common factors in the second step. Also, data containing sub-panels or block structure is a common situation where weak factors arise. \cite{hallin2011dynamic} developed method to identify and estimate joint and block-specific common factors among different panels. Similar result can be achieved by using constrained vector factor model in \cite{Tsai-Tsay-2010}. Data containing sub-panels can also be cast into matrix observations by putting sub-panels as columns. The column spaces of loadings can be divided into subspaces that correspond to the joint and block-specific common factors. However, more sophisticated estimation procedures need to be developed to exclude overlap of the column spaces. The constrained matrix factor model provides building blocks for future research on combining constraints to represent different structures and on devising estimation procedures. 

\clearpage
\bibliographystyle{agsm}
\bibliography{cmfm}

\end{document}